\documentclass[11pt]{article}

\usepackage{booktabs} 

\usepackage{geometry}
\usepackage{graphicx}
\usepackage{amsmath}
\usepackage{amssymb}
\usepackage[all]{xy}

\newtheorem{theorem}{Theorem}[section]
\newtheorem{corollary}[theorem]{Corollary}
\newtheorem{lemma}[theorem]{Lemma}
\newtheorem{proposition}[theorem]{Proposition}
\newtheorem{claim}[theorem]{Claim}
\newtheorem{definition}[theorem]{Definition}
\newtheorem{example}[theorem]{Example}
\newtheorem{observation}[theorem]{Observation}

\def\squarebox#1{\hbox to #1{\hfill\vbox to #1{\vfill}}}
\newcommand{\qed}{\hspace*{\fill}\vbox{\hrule\hbox{\vrule\squarebox{.667em}\vrule}\hrule}\smallskip}
\newenvironment{proof}{\noindent{\bf Proof:~~}}{\(\qed\)}

\newcommand{\xhdr}[1]{\paragraph{\bf #1}}

\newcommand{\shortversion}[1]{}
\newcommand{\fullversion}[1]{#1}

\begin{document}
\title{Combinatorial Auctions with Endowment Effect}        


%
%

\author{Moshe Babaioff \and Shahar Dobzinski \and Sigal Oren}

\maketitle

\begin{abstract}
We study combinatorial auctions with bidders that exhibit endowment effect. In most of the previous work on cognitive biases in algorithmic game theory (e.g., [Kleinberg and Oren, EC'14] and its follow-ups) the focus was on analyzing the implications and mitigating their negative consequences. In contrast, in this paper we show how in some cases cognitive biases can be harnessed to obtain better outcomes.

Specifically, we study Walrasian equilibria in combinatorial markets. 
It is well known that Walrasian equilibria exist only in limited settings, e.g., when all valuations are gross substitutes, but fails to exist in more general settings, e.g., when the valuations are submodular.
We consider combinatorial settings in which bidders exhibit the \emph{endowment effect}, that is, their value for items increases with ownership.

Our main result shows that when the valuations are submodular, even a mild degree of endowment effect is sufficient to guarantee the existence of Walrasian equilibria.
In fact, we show that in contrast to Walrasian equilibria with standard utility maximizing bidders -- in which the equilibrium allocation must be efficient -- when bidders exhibit endowment effect any \emph{local} optimum can be an equilibrium allocation.

Our techniques reveal interesting connections between the LP relaxation of combinatorial auctions and local maxima. We also provide lower bounds on the intensity of the endowment effect that the bidders must have in order to guarantee the existence of a Walrasian equilibrium in various settings.
\end{abstract}


\section{Introduction}

Research in algorithmic mechanism design typically assumes that bidders are utility maximizers, i.e., they maximize their value for the chosen alternative minus their payment. However, empirical evidence from behavioral economics suggests that this assumption is often inaccurate. In practice, individuals tend to exhibit different cognitive biases that are not captured by the classic model of utility maximization. For example, the price tag might affect their value for a bottle of wine (irrational value assessment); or they may
attribute higher values to items 
once
they own them 
(endowment effect).

A recent line of work in algorithmic game theory (Kleinberg and Oren \cite{kleinbergTime} and their follow-ups 
\cite{present-bias-tang, Kraft-time,Gravin:present-bias, kleinberg-soph, Kleinberg-mult-bias, Kraft-Uncertainty, kraft-penalties}) mathematically models and analyzes the behavior of agents that exhibit cognitive biases in planning settings. In contrast, in this paper we initiate the study of cognitive biases in the central setting of algorithmic mechanism design: combinatorial auctions.\footnote{In the context of algorithmic mechanism design, the only paper that incorporates behavioral assumptions that we are aware of is \cite{ghosh15-ec}. The paper applies prospect theory to a crowdsourcing setting.} Furthermore, this line of work mostly focused on analyzing the implications of cognitive biases and mitigating their \emph{negative} consequences. In this paper we take a different approach and show that in some settings cognitive biases can be harnessed to obtain \emph{better} outcomes.

Walrasian equilibrium is the traditional concept of market equilibrium. Roughly speaking, a market is in Walrasian equilibrium if there exists a price for each item such that every bidder is allocated his most profitable bundle and all items are allocated. Walrasian equilibria are attractive since they are simple and since their pricing structure is easy to understand. Not only that, it is also guaranteed that any allocation of such an equilibrium must be efficient (the ``First Welfare Theorem'').

A major downside of Walrasian equilibria is that often they do not exist at all. 
In the context of combinatorial auctions it is known that a Walrasian equilibrium is guaranteed to exist when the valuations are gross substitutes \cite{gul1999walrasian}. This relatively small class of valuations is in a (formal) sense the maximal one guaranteeing existence \cite{gul1999walrasian}. In particular, a Walrasian equilibrium is not guaranteed to exist if the valuations are submodular.

In this paper we consider bidders in combinatorial auctions that exhibit an endowment effect. 
Our main contribution is utilizing the cognitive biases of the bidders to extend the valuations classes for which a Walrasian equilibrium exists: if bidders exhibit a mild endowment effect then a Walrasian equilibrium is guaranteed to exist even when the valuations are submodular. 


Before formally presenting our model and results, we now discuss the endowment effect.

\xhdr{Endowment Effect.}
Do owning items makes them more valuable to us? The Nobel Laureate economist Richard Thaler coined the term endowment effect \cite{thaler1980toward} to explain situations in which the mere possession of an item makes it more valuable. An experiment by Knetsch \cite{knetschEndowment} provides a stark illustration: two groups of students received goods in return for filling out a questionnaire. One group received mugs and the other chocolate bars. Next, the students were given the option to swap items. Since the items were allocated to the students randomly, we expect that about half of them would have received the less desirable item and would like to switch. In contrary to this logic,  90\% of the students in each group decided to keep their endowed item. 

A classic experiment by Kahneman, Knetsch and Thaler \cite{kahneman1990experimental} attempts to quantify the endowment effect. In this experiment half of the students in a Law and Economics class at Cornell University received a $\$6$ mug. After examining the mug, the students who received the mug were asked at which price are they willing to sell the mug (willingness to accept). The students who did not receive the mug were asked for the price they are willing to pay for the mug (willingness to pay). As in the Knetch's experiment, one could have expected that about $50\%$ of the mug owners could be matched with a buyer that values the mug more than they do. However, in the experiment only $20\%$ of the mug owners sold their mug. 
Moreover, there was a significant gap between the median price sellers were willing to sell at ($ \$5.25$), and the median price that buyers where willing to buy at ($\$2.25$). 

The endowment effect can be thought of as a special case of the status quo bias \cite{samuelson1988status}, according to which individuals have a strong preference for the current state (i.e., a mug owner prefers to stay a mug owner). Kahneman et al. \cite{kahneman1990experimental,kahneman1991anomalies} claim that loss aversion from prospect theory \cite{kahneman-prospect} is the source of the endowment effect (i.e., for a mug owner exchanging the mug for a chocolate bar entails \emph{losing} the mug). K{\H{o}}szegi and Rabin \cite{kHoszegi2006model} incorporate reference points and expectations into prospect theory. Their framework suggests that the endowment effect will only appear in settings in which individuals do not expect to trade.\footnote{See \cite{list2003does,engelmann2010reconsidering} for some experimental work supporting this prediction.} Masatlioglu and Ok \cite{masatlioglu2005rational} propose a model in which the individuals make a rational choice on a constrained set of alternatives defined according to the status quo.


Several papers study endowment effect in auction settings (e.g., \cite{knetsch2001endowment,heyman2004auction}). However, apart from \cite{burson2013multiple} that discusses the endowment effect with multiple identical goods, all considered the single item setting, leaving open the question of incorporating the endowment effect into the more general combinatorial auctions setting. 

\xhdr{The Model.} 

We consider a combinatorial auction with $n$ bidders and a set $M$ of $m$ items. Each bidder $i$ has a valuation function $v_i:2^M\rightarrow \mathbb{R}_+$ that specifies his value for every bundle. The valuation function is non-decreasing and normalized ($v_i(\emptyset)=0$). 

Our modeling of the endowment effect is driven by the essence of this effect: an individual attributes an item a higher value after owning it; his value for items that he does not own remains the same. Formally, bidder $i$ who previously valued item $a$ at $v_i(a)$ will now value it at $\alpha_i \cdot v_i(a)$ for some $\alpha_i\geq 1$.\footnote{The case $0\leq \alpha_i < 1$ is also of conceptual interest. In particular, it models situations of the type ``the grass is always greener on the other side'', when a person values items less once he owns them. The analysis of this case is technically simpler and is provided in \fullversion{Appendix \ref{appendix-small-alpha}}\shortversion{in the full version}.} Similarly, when turning to the more general setting of combinatorial auctions we define the value of a bundle $S_i$ that bidder $i$ owns as $\alpha_i \cdot v_i(S_i)$. Note that it is insufficient to define the new value of bidder $i$ only for the bundle $S_i$, we also need to define how the value of any other bundle $T$ changes when bidder $i$ owns $S_i$. We follow a similar line of reasoning: the value of the items in $S_i\cap T$ is multiplied by $\alpha_i$, while the marginal contribution of the remaining items $T-S_i$ remains the same. Formally, we define the endowed valuation of player $i$ with $\alpha_i$ who is endowed with
bundle $S_i$ to be:\footnote{One can also give a ``classic'' interpretation that does not involve cognitive biases to this transformation: consider an environment with transaction costs, e.g., a company that acquires a new location might be subject to property taxes and/or have to invest in new infrastructure. Similarly, a shop that moves to a new location might have to start a costly advertising campaign to inform its costumers about the move. 
In all these cases after owning their resources the preference of the agents to the new status quo allocation increases.}
\begin{align*}
\forall T \subseteq M,~~~~ v^{S_i,\alpha_i}_i(T) = \alpha_i \cdot v_i(S_i\cap T) + v_i(T - S_i | S_i \cap T)
\end{align*}
where $v(X|Y)=v(X\cup Y)-v(Y)$ denotes the marginal contribution of the bundle $X$ given $Y$.

Recall that a Walrasian equilibrium consists of an allocation $(S_1,\ldots, S_n)$ and prices $(p_1,\ldots, p_m)$ such that:
\begin{enumerate}
\item For each player $i$, $v_i(S_i)-\sum_{j\in S_i} p_j\geq v_i(T)-\sum_{j\in T} p_j$, for every bundle $T \subseteq M$.
\item If $j\notin \cup_iS_i$ then $p_j=0$.
\end{enumerate}

Several works attempted to relax the definition of a Walrasian Equilibrium in order to ensure existence
\cite{feldman2016combinatorial, fu_conditional}\footnote{Feldman et al. \cite{feldman2016combinatorial} use bundle prices instead of individual item prices and do not require market clearance; Fu et al. \cite{fu_conditional} define a conditional equilibrium where each player does not wish to add any items to the bundle he was allocated. Also related are works on simultaneous first \cite{DBLP:conf/sigecom/HassidimKMN11} and second price \cite{christodoulou2008bayesian} auctions.}. Our approach is different; Given that the endowment effect changes the valuations anyways, we are simply interested in a Walrasian equilibrium with respect to the endowed valuations: an allocation $S$ and non-negative item prices $(p_1,\ldots, p_m)$ form an \emph{$(\alpha_1,\ldots, \alpha_n)$-endowed equilibrium} in an instance $(v_1,\ldots, v_n)$ if they constitute a Walrasian equilibrium in the instance $(v^{S_1,\alpha_1}_1,\ldots, v^{S_n,\alpha_n}_n)$. 

When $\alpha=\alpha_1=\ldots=\alpha_n$ we use the compact notation \emph{$\alpha$-endowed equilibrium}. Let $S$ be the allocation and $(p_1,\ldots, p_m)$ be the prices of some  $(\alpha_1,\ldots, \alpha_n)$-endowed equilibrium. Observe that $S$ and $(p_1,\ldots, p_m)$ 
also form an $(\alpha, \alpha, \ldots,\alpha)$-endowed equilibrium, for $\alpha=\max_i\alpha_i$. The reason is that when $S_i$ maximizes the profit of $v_i^{S_i,\alpha_i}$, it also maximizes the profit of $v_i^{S_i,\alpha'}$, for every $\alpha'>\alpha_i$. Thus, in the rest of the paper we let $\alpha = \max_i \alpha_i$ and focus on $\alpha$-endowed equilibria.


Note that for $\alpha=1$ we recover the classic notion of Walrasian Equilibrium while, roughly speaking, as $\alpha$ goes to infinity the players insist more and more on keeping the bundle they were allocated.\footnote{When $\alpha$ is very large this is 
	similar to conditional equilibrium \cite{fu_conditional} as the bidders do not want to discard any item with marginal value greater than $0$.} Throughout the paper we say that $\alpha$ \emph{supports} an allocation $S$ if there exist item prices that together with $S$ form an $\alpha$-endowed equilibrium.

\xhdr{Results.}  Analyzing Walrasian equilibria in the context of endowed valuations brings with it a natural question: when does an endowed equilibrium exist? The simple answer is \emph{always}; we show that in any instance there exists an $\alpha>1$ and an allocation for which an $\alpha$-endowed equilibrium exists. However, this answer completely misses the point as the value of $\alpha$ for which an endowed equilibrium exists might be huge. Recall that the value of $\alpha$ corresponds to the intensity of the endowment effect, thus we expect $\alpha$-endowed equilibria that arise in practice to have small value of $\alpha$ (in the experiment of \cite{kahneman1990experimental} mentioned above, for example, it seems like $\alpha$ was about $2$). 
This leads us to the definition of the endowment gap -- the minimal value of $\alpha$ for which a Walrasian Equilibrium is guaranteed to exist. Thus, the main question that we ask in this paper is \emph{what is the endowment gap for different valuation classes?}


Our main result shows  that for submodular valuations there is always an allocation that is supported by $\alpha=2$. That is, \emph{in combinatorial auctions with submodular valuations the endowment gap is at most $2$. }

The proof of the theorem is constructive: an allocation is a local maximum if the welfare of the allocation cannot be improved by moving a single item from one player to another. We show that in any instance of combinatorial auctions with submodular valuations \emph{any local maximum} is supported by $\alpha=2$. In contrast, the First Welfare Theorem asserts that the allocation in a Walrasian equilibrium is a \emph{global} optimum. Hence, a local maximum that is not a global maximum can be part of a $2$-endowed equilibrium, but cannot be a part of a Walrasian equilibrium.

Our work reveals interesting connections between the integrality gap of the linear program relaxation for combinatorial auctions and the endowment gap. Nisan and Segal \cite{nisan2006communication} show that a Walrasian equilibrium exists if and only if the integrality gap of a natural relaxation of the LP for combinatorial auctions (the ``configuration LP'') is $1$, in other words, if and only if there is an optimal integral solution to the LP. This in turn implies that an (integral) allocation $S$ is supported by $\alpha$ if and only if {it is an optimal solution {of} the (fractional) LP with respect to the perturbed valuations}. 

An equivalent geometric interpretation is the following: consider the polytope defined by a given instance of a combinatorial auction. 
With submodular valuations, there might not be any optimal integral solution on a vertex of the polytope. 
Fix some allocation and consider the change to the objective function of the LP as $\alpha$ grows (changing the corresponding endowed valuations.)
As $\alpha$ grows, the direction changes, rotating towards the direction of the endowed allocation. 

The allocation can be supported by the minimal value of $\alpha$ {(if exists)} for which the (integral) endowment allocation becomes an optimal solution to the LP. 
For submodular valuations, {our result shows that this happens quickly:
when the allocation is a local maximum then a value of $\alpha$ of only $2$ suffices.}

It is not hard to see that the endowment gap is at least the integrality gap.
However, note that the endowment gap is typically strictly larger. More generally, by analyzing the LP we give a precise definition for the minimal value of $\alpha$ required to support a given allocation. Roughly speaking, this minimal value of $\alpha$ that supports an allocation $S$ must be bigger than some combination of the value of any fractional solution plus the ``intersection'' of this fractional solution with $S$. See a formal treatment in Section \ref{subsec-integrality}.

In fact, the LP point of view provides some interesting implications of our main result beyond the realm of endowed valuations: it is implicit in previous work that a local maximum in combinatorial auctions with submodular bidders provides a $2$ approximation to the welfare maximizing solution. One implication of our main result is that a local maximum provides a $2$ approximation even with respect to the optimal \emph{fractional} social welfare. More generally, the equilibrium allocation in an $\alpha$-endowed equilibrium provides an $\alpha$ approximation to the optimal fractional social welfare.


We also show that our analysis of submodular valuations is tight in the following sense: there is an instance in which any local maximum requires $\alpha\geq 2$ to be supported (but there are other allocations that can be supported by a smaller value of $\alpha$). We also show that there is an instance with just $2$ bidders in which \emph{every} allocation requires $\alpha\geq 1.5$ to be supported. We thus conclude that the endowment gap for submodular valuations is between $1.5$ and $2$.

%


What about classes of valuations, like XOS or subadditive?  Here we give a definite negative result: for every $\alpha>1$ there exists an instance with just two bidders with identical XOS valuations and three identical items in which the endowment gap is $\alpha' > \alpha$.

\xhdr{Open Questions.} This work models the endowment effect in combinatorial auctions and analyzes the endowment gap. Our main finding is that every local maximum can be supported by $\alpha=2$ when the valuations are submodular and that the endowment gap for submodular valuations is at least $1.5$. An obvious open question is to close this gap.  A related open question is to analyze the endowment gap of subclasses of submodular valuations. For example, for budget additive valuations we are able to show that the endowment gap is also at least $1.5$, but we are unable to prove that the endowment gap is strictly smaller than $2$ even for this restricted class.

The focus of this paper is on characterizing the existence of $\alpha$-endowed equilibria. An interesting follow-up question is to understand the ``computational endowment gap'': the minimal value of $\alpha$ for which an $\alpha$-endowed equilibrium can be efficiently computed. One would hope that for submodular valuations a local maximum can be efficiently computed. Unfortunately, we show that there are both communication and computation hurdles in finding a local maximum:
finding a local maximum in combinatorial auctions with submodular bidders requires an exponential number of value queries. Moreover, we present a family of succinctly represented submodular valuations for which finding a local optimum is PLS hard.
On the somewhat more positive side, we do know how to find with only polynomially many value queries an allocation (not necessarily a local maximum) that can be supported by some $\alpha>1$. However, this value of $\alpha$ is typically huge. Are there other allocations that can both be efficiently computed and supported by a small value of $\alpha$? Remarkably, we are unable to provide an answer for this question for any reasonable value of $\alpha$, not even for a restricted class like budget additive valuations. In fact, we do not know if finding a local maximum when the valuations are budget additive is computationally hard. This question might be of independent interest, regardless of the specific application to the endowment gap.

Another interesting question is to devise natural auction methods that end up with an endowed equilibrium. If the valuations are gross-substitutes, then there are natural ascending auctions that end up with a Walrasian equilibrium (e.g., \cite{gul1999walrasian}). Are there natural ascending auctions that end up with an endowed equilibrium when the valuations are submodular? One question that might arise while developing such ascending auctions is to understand the extent to which bidders exhibit an endowment effect with respect to items that are temporarily assigned to them during the auction and take this temporary endowment effect -- if exists -- into account.

Finally, a natural measure of how far the market is from equilibrium suggests itself. Recall that a valuation $v$ is $c$-approximated by a valuation $v'$ if for every bundle $S$, $\frac{v'(S)}{c}\leq v(S) \leq v'(S)$. Given valuations $v_1,\ldots, v_n$, define the ``distance to equilibrium'' as the minimal $c$ for which there exist $v'_1,\ldots, v'_n$ such that each $v'_i$ $c$-approximates $v_i$ and the instance $v'_1,\ldots, v'_n$ admits a Walrasian equilibrium. Since the endowment gap for
submodular valuations is at most $2$, this means that the distance to equilibrium of any instance with submodular valuations is at most $2$. It will be interesting to see if this result can be improved, e.g., maybe by using valuations that are not endowed valuations. Similarly, what is the distance to equilibrium of instances with subadditive or XOS valuations?

\section{The Model}

There are $n$ players and a set $M$ of $m$ goods, each agent $i$ has a combinatorial valuation function $v_i: 2^M\rightarrow \mathbb{R}_+$.  We assume that for each player $i$, $v_i$ is monotone ($S\subseteq T$ implies $v_i(S)\leq v_i(T)$) and normalized ($v_i(\emptyset)=0$).
 We use the notation $v_i(T|S)$ to denote $v_i(S\cup T)-v_i(S)$, the marginal value of $T$ given $S$. 


Each player $i$ has a parameter $\alpha_i$ that measures the intensity of his endowment effect. Specifically, if player $i$ is endowed with a bundle $S_i$, then his valuation function is 
\begin{align*}
v_i^{S_i,\alpha_i}(T)&=\alpha_i\cdot v_i(S_i\cap T)+v_i(T-S_i|S_i\cap T) \\
	&= v_i(T) + (\alpha_i-1) \cdot v_i(S_i \cap T)
\end{align*}
We will use both expressions interchangeably.

%
%
%

%


An allocation $S=(S_1,\ldots,S_n)$ assigns to each agent $i$ a set $S_i$ such that for every $i\neq j$, $S_i\cap S_j=\emptyset$.
An allocation $(S_1,\ldots,S_n)$ and (non-negative) item prices $(p_1,\ldots, p_m)$ constitute an \emph{$(\alpha_1,\ldots,\alpha_n)$-endowed equilibrium} if:
\begin{enumerate}
\item For each player $i$, $v^{S_i,\alpha_i}_i(S_i)-\sum_{j\in S_i} p_j\geq v^{S_i,\alpha_i}_i(T)-\sum_{j\in T} p_j$, for every bundle $T\subseteq M$.
\item If $j\notin \cup_i S_i$ then $p_j=0$.
\end{enumerate}
When $\alpha=\alpha_1=\ldots=\alpha_n$ we shorten the name to \emph{$\alpha$-endowed equilibrium}. Let $S$ be the allocation and $(p_1,\ldots, p_m)$ be the prices of some  $(\alpha_1,\ldots, \alpha_n)$-endowed equilibrium. Observe that $S$ and $(p_1,\ldots, p_m)$ 
also form an $(\alpha, \alpha, \ldots,\alpha)$-endowed equilibrium, for $\alpha=\max_i\alpha_i$. 
The reason for this is that when $S_i$ maximizes the profit of $v_i^{S_i,\alpha_i}$, it also maximizes the profit of $v_i^{S_i,\alpha'}$, for every $\alpha'>\alpha_i$. Thus, from this point onwards we let $\alpha = \max_i \alpha_i$ and focus on $\alpha$-endowed equilibrium. 

An allocation $(S_1,\ldots,S_n)$ is \emph{supported by $\alpha$} if there exist prices $(p_1,\ldots, p_m)$ such that the prices and the allocation form an $\alpha$-endowed equilibrium.
In particular, in every instance in which a Walrasian equilibrium exists (e.g., every instance in which the valuation functions are gross substitutes), we obviously have an endowed equilibrium supported by $\alpha=1$. 
In instances where a Walrasian equilibrium does not necessarily exist, we will be looking for the minimal value of $\alpha$ for which an $\alpha$-endowed equilibrium exists. 

The conceptually and technically interesting regime is when $\alpha>1$, that is, a player assigns higher value for items in their endowment, but see \fullversion{Appendix \ref{appendix-small-alpha}}\shortversion{the full version} for the regime $0\leq\alpha<1$.

In this work we are interested in the following valuations classes (each class is contained in its successor in the list):
\begin{itemize}
\item \emph{Additive valuations:} A valuation $v$ is additive if for every $S$, $v(S)=\sum_{j\in S}v(\{j\})$.
\item \emph{Budget additive valuations:} A valuation $v$ is budget additive if there exists $b$ such that for every $S$, $v(S)=\min\{b,\sum_{j\in S}v(\{j\})\}$.
\item \emph{Submodular valuations:} a valuation $v$ is submodular if for every $S,T$, $v(S)+v(T)\geq v(S\cup T)+v(S\cap T)$.
\item \emph{XOS valuations:} a valuation is XOS if there exists additive valuations $\{a_1,\ldots,a_l\}$ such that for every bundle $S$, $v(S)=\max_{1 \leq k \leq l}a_k(S)$. 
\item Subadditive valuations: a valuation $v$ is subadditive if for every $S,T$, $v(S)+v(T)\geq v(S\cup T)$.
\end{itemize}

\section{The Endowment Gap}


Consider some instance of a combinatorial auction with $n$ players with valuations $v_1,\ldots, v_n$ and a set $M$ of $m$ items. For a given instance, the \emph{endowment gap} is, roughly speaking, the minimal value of $\alpha$ for which an $\alpha$-endowed equilibrium exists in that instance.
We are interested in proving bounds on the value of $\alpha$ for which an $\alpha$-endowed equilibrium exists for all instances of classes of combinatorial  valuations (e.g. submodular, XOS, subadditive).

\begin{definition}
	The \emph{endowment gap of an instance $(v_1,\ldots, v_n)$ with respect to an allocation $A=(A_1,\ldots, A_n)$}, denoted $\mathcal{G}^A(v_1,\ldots, v_n)$, is the infimum of the values of $\alpha$ that support $A$. 
\end{definition}
We naturally extend the definition of an endowment gap to an instance and to a class of valuations:
\begin{definition}
~
\begin{itemize}
\item The \emph{endowment gap of an instance $(v_1,\ldots, v_n)$} is the minimum value, over every allocation $S$, of the endowment gap with respect to the allocation $S$: 
$\min_{S}\mathcal{G}^S(v_1,\ldots, v_n)$.

\item 
The \emph{endowment gap of a class of valuations $\mathcal V$} is the supremum over the endowment gaps over all {valuations profiles in which each valuation belongs to the class $\mathcal V$}: 
$$
\sup_ {(v_1,\ldots,v_n)\in \mathcal V^n} \min_{S}\mathcal{G}^S(v_1,\ldots, v_n)
$$ 
\end{itemize}
\end{definition}

Next, we provide a simple characterization that shows that any allocation $A$ that its social welfare cannot be improved by reallocating items that do not contribute to the social welfare of $A$, can be supported  by some $\alpha>1$.

\begin{definition}
Let $S=(S_1,\ldots, S_n)$ be some allocation. For every item define $q_j$ to be the marginal contribution of item $j$ to the allocation $S$ as follows: if there exists $i$ such that $j\in S_i$, let $q_j=v_i(j|S_i-\{j\})$ . For every item $j$ that is not allocated, let $q_j=0$. 
Let $Z=\{j|q_j=0\}$. $S$ is \emph{maximal} if for every player $i$, $v_i(S_i\cup Z)=v_i(S_i)$.
\end{definition}
Note that in particular, in a maximal allocation there is no bidder with zero marginal value for an item, for which some other bidder has positive marginal value given his set.

\begin{proposition}\label{prop-maximal}\ 
\begin{enumerate}
\item Every maximal allocation $S=(S_1,\ldots, S_n)$ can be supported by some $\alpha>0$. Furthermore, for a given allocation, we can find some $\alpha$ that supports it and the prices with $poly(n,m)$ value queries.
\item If an allocation $S=(S_1,\ldots, S_n)$ is not maximal then there is no $\alpha\geq 1$ that supports it.
\end{enumerate}
\end{proposition}
\fullversion{
\begin{proof}
Let $Q^+=M-Z$ (the set of items $j$ with positive marginal contribution.
Let $\overline{OPT}$ be some upper bound on the value of the optimal solution (e.g., $n\cdot \max_iv_i(M)$ or simply $OPT$ if computational considerations are irrelevant). We use the following prices to support $S$: $p_j=2\cdot \overline{OPT}$ for $j\in Q^+$ and $p_j=0$ for $j\in Z$. We will show that $\alpha=\frac {20m\cdot \overline{OPT}} {\min_{j\in Q^+}q_j}$ and these prices form an endowed equilibrium.
	
We first show that a player will never drop items that are in his set but not in $Z$. I.e., if we denote by $T$ some bundle that maximizes the profit of player $i$ then $S_i- Z\subseteq T$. Specifically, we will show that $v_i(S_i\cap T)=v_i(S_i)$. Observe that $v_i(S_i\cap T)=v_i(S_i)$ indeed implies that $S_i- T\subseteq Z$, since otherwise there exists $j\in Q^+$ such that $j\in S_i- T$. We then get a contradiction since $v_i(S_i)>v_i(S_i- \{j\})\geq v_i(S_i\cap T)$, where the first inequality is because $j\in Q^+$ and the second one is due to the monotonicity of $v_i$ and the fact that $j\notin T$.

We now show that the profit of player $i$ from a bundle $T$ such that $v_i(S_i\cap T)<v_i(S_i)$ is strictly smaller than the profit of the bundle $T\cup S_i$.
	\begin{align*}
		v_i^{S_i,\alpha}(T\cup S_i)-\sum_{j'\in T\cup S_i}p_{j'}&=
		\alpha\cdot  v_i(S_i)+  v_i(T - S_i|S_i)-\sum_{j'\in S_i}p_{j'}-  \sum_{j'\in T - S_i}p_{j'} \\
		&= (\alpha-1)\cdot  v_i(S_i)+  v_i(T)-\sum_{j'\in S_i}p_{j'}-  \sum_{j'\in T- S_i}p_{j'} \\
		&= (\alpha-1)\cdot  v_i(S_i\cap T)+ (\alpha-1)\cdot  v_i(S_i- T|S_i\cap T)+ v_i(T)-\sum_{j'\in T}p_{j'}-  \sum_{j'\in S_i - T}p_{j'}
	\end{align*}
Observe that
$v_i(S_i- T|S_i\cap T) = v_i(S_i)- v_i(S_i\cap T) \geq q_j\geq  \min_{j\in Q^+}q_j$ (which holds by our discussion above since $v_i(S_i)-v_i(S_i\cap T)\geq q_j$ as $j\notin T$ and since $j\in Q^+$) and $(\alpha-1)\cdot \min_{j\in Q^+}q_j=(\frac {20m\cdot \overline{OPT}} {\min_{j\in Q^+}q_j}-1)\cdot \min_{j\in Q^+}q_j=20m\cdot\overline{OPT}-\min_{j\in Q^+}q_j\geq 19m\cdot\overline{OPT}$. Thus, we have that:

	\begin{align*}
		v_i^{S_i,\alpha}(T\cup S_i)-\sum_{j'\in T\cup S_i}p_{j'}
		&\geq (\alpha-1)\cdot  v_i(S_i\cap T)+   19m\cdot \overline{OPT}+ v_i(T)-\sum_{j'\in T}p_{j'}-  2m\cdot \overline{OPT} \\
		&>(\alpha-1)\cdot  v_i(S_i\cap T)+  v_i(T)-\sum_{j'\in T}p_{j'} \\
		&=v_i^{S_i,\alpha}(T)-\sum_{j'\in T}p_{j'}
	\end{align*}

We next show that the demand of any player is always a subset of  $S_i\cup Z$. 
That is, if $T$ is some bundle that maximizes the profit of player $i$ then $T\subseteq S_i\cup Z$. 
Let $R=T-  (S_i\cup Z)$ and suppose towards contradiction that $R\neq \emptyset$. We claim that the profit of the bundle $T-  R$ for player $i$ is strictly higher than that of $T$:
\begin{align*}
v_i^{S_i,\alpha}(T)-\sum_{j\in T}p_{j}&=
v_i^{S_i,\alpha}((T-  R) \cup R)-\sum_{j\in T}p_{j}\\
&=v_i^{S_i,\alpha}(T-  R) + v_i(R| (T-  R))-\sum_{j\in T-  R}p_{j} - \sum_{j\in R}p_{j}\\
&\leq v_i^{S_i,\alpha}(T-  R) + \overline{OPT}-\sum_{j\in T-  R}p_{j} - 2\overline{OPT}\\
&< v_i^{S_i,\alpha}(T-  R) -\sum_{j\in T-  R}p_{j}
\end{align*}

Together with our first observation, we have that a bundle $T$ that maximizes the profit of player $i$ must satisfy $T=S_i\cup Z'$, for some $Z'\subseteq Z$. 
Recall that $v_i(S_i\cup Z')\leq v_i(S_i \cup Z)=v_i(S_i)$ and that the price of every $j\in Z$ is $p_j=0$ and we get that the profit from $T$ is exactly the profit of $S_i$. Thus, for every player $i$, $S_i$ is a profit maximizing bundle, as needed. Finally, notice that to compute the prices and some $\alpha>0$ we only need to find some upper bound on OPT (as noted above, computing $n\cdot \max_iv_i(M)$ takes $n$ value queries) and the marginal contribution of every item $j$ ($2$ queries for each item).

For the second part of the proof, consider an allocation $S$ that is not maximal. We will see that for any $\alpha$ there are no prices that $\alpha$-support this allocation. The proof is based on the simple observation that in any endowed equilibrium the price of every item $j\in Z$ must be $0$: this must be the case by definition for every item $j$ that is not allocated. If item $j\in S_i$ and the price of $j\in Z$ is positive, then the profit of player $i$ from the bundle $S_i- \{j\}$ is greater than his profit from the bundle $S_i$. 

Since $S$ is not maximal, there is some player $i$ such that $v_i(S_i\cup Z) > v_i(S_i)$. Using the simple observation, the price of every item $j\in Z$ is $0$, thus the profit of player $i$ from the bundle $v_i(S_i\cup Z)$ is strictly larger than the profit from the bundle $v_i(S_i)$. Therefore, $S$ cannot be $\alpha$-supported, for any $\alpha$.
\end{proof}
} 

\fullversion{In Appendix \ref{sec:support-instance}}\shortversion{The proof can be found in the full version. In addition, in the full version of the paper} we use this characterization to show that not only in every instance there is an allocation that can be supported, but there is even \emph{some welfare maximizing allocation} that can be $\alpha$-supported by some $\alpha>1$. The caveat is that the $\alpha$ that we guarantee might be huge. In Claim \ref{claim-endowment-vs-integarlity} we give the exact value of the minimal $\alpha$ that supports an allocation $S$. However, even this precise characterization might result in large values of $\alpha$. This is no coincidence: Proposition \ref{prop-XOS} shows that for every fixed $\alpha$ there is an instance for which the endowment gap is strictly bigger than $\alpha$. Thus in most of this paper we restrict our attention to specific classes of valuations,
aiming to find bounds on the endowment gap that hold for all instances in the class and, ideally, find prominent classes of valuations for which the gap is small. 

\subsection{The Endowment Gap and the LP Relaxation}\label{subsec-integrality}
\label{sec:int-gap}
In this section we explore the connections between the endowment gap and the following linear program relaxation for combinatorial auctions:

\vspace{0.1in}\noindent  \emph{Maximize:} 
$\sum_{i=1}^n \sum_{S\subseteq M} x_{i,S} \cdot v_i(S)$

\noindent  \emph{Subject to:}
  \begin{itemize}
      \item For each item $j$: $\sum_{i=1}^n \sum_{S\subseteq M | j \in S}~x_{i,S}\leq 1$.
      \item for each bidder $i$: $\sum_{S \subseteq M}~x_{i,S}\leq 1$.
      \item for each $i$, $S$: $x_{i,S}\geq 0$.
  \end{itemize}
This linear program is tightly connected to the notion of \emph{Walrasian equilibrium}:
\begin{theorem}[\cite{nisan2006communication}]
For every instance $(v_1,\ldots, v_n)$, there exists
a Walrasian Equilibrium in the instance $(v_1,\ldots, v_n)$ if and only if the integrality gap of the above linear program is $1$.
Moreover, an integral allocation is the allocation of some Walrasian Equilibrium if and only if it is an optimal solution to the LP.
\end{theorem}
When considering endowed valuations $v^{A_1,\alpha}_1,\ldots,v^{A_n,\alpha}_n$ the theorem immediately implies an analogous result for $\alpha$-endowed equilibrium:
\begin{corollary}\label{cor-integralitygap-is-1}
In an instance $(v_1,\ldots,v_n)$ an allocation $A=(A_1,\ldots, A_n)$ is $\alpha$-supported if and only if  
$A$ is an optimal solution to the LP of the instance $(v^{A_1,\alpha}_1,\ldots,v^{A_n,\alpha}_n)$  
(implying in particular that the integrality gap of the latter instance is $1$.) 
Furthermore, by the first welfare theorem we get that $A$ is welfare maximizing with respect to $v^{A_1,\alpha}_1,\ldots,v^{A_n,\alpha}_n$.
\end{corollary}

We can also relate the welfare of supported allocations to that of fractional allocations. In Subsection \ref{subsec-implication} we use the next corollary to improve the bounds on the welfare guaranteed by local maxima in combinatorial auctions with submodular valuations. 
\begin{corollary} \label{cor-fractional-approximation}
In an instance $(v_1,\ldots,v_n)$, if an allocation $A=(A_1,\ldots, A_n)$ is $\alpha$-supported then it provides an $\alpha$-approximation to the {maximum \emph{fractional} welfare of the instance $(v_1,\ldots, v_n)$}. 
\end{corollary}
\begin{proof}
Let $\{x_{i,S}\}$ be {some} fractional solution of the LP {of the instance} $(v^{A_1,\alpha}_1,\ldots,v^{A_n,\alpha}_n)$. 
As the allocation $A=(A_1,\ldots, A_n)$ is $\alpha$-supported, by Corollary  \ref{cor-integralitygap-is-1}  and the definition of endowed valuations we have that:
\begin{align*}
\alpha \sum_{i=1}^n v_i(A_i) \geq \sum_{i=1}^n \sum_{S\subseteq M} x_{i,S} \cdot v_i^{A_i,\alpha}(S) \geq \sum_{i=1}^n \sum_{S\subseteq M} x_{i,S} \cdot v_i(S)
\end{align*}
In particular, this holds for the welfare maximizing fractional solution of the instance $(v_1,\ldots, v_n)$, implying that $A$ provides an $\alpha$-approximation to the value of that fractional allocation, as needed.
\end{proof}

As we will see next,
the endowment gap has some interesting and useful connections to the integrality gap. For our first application, recall that Proposition \ref{prop-maximal} shows that an allocation can be supported by some $\alpha$ if and only if it is maximal. We now use the connection to the LP to determine the minimal value of $\alpha$ that can support a maximal allocation.

\begin{claim}\label{claim-endowment-vs-integarlity}
Let $A=(A_1,...,A_n)$ be some allocation. Given a fractional solution $\{x_{i,S}\}$ to the LP, define 
$$\psi_{A, \{x_{i,S}\}} =\sum_{i=1}^n \sum_{S \subseteq M} x_{i,S} \cdot v_i(S \cap A_i)
$$

Suppose that $A$ is supported by $\alpha$. Then,
\begin{enumerate}
\item\label{claim-psi-first-part} For every fractional solution $\{x_{i,S}\}$, 
$  \alpha  \cdot \sum_{i=1}^n v_i(A_i) \geq  \sum_{i=1}^n \sum_{S \subseteq M} x_{i,S} \cdot v_i(S) +(\alpha-1) \psi_{A, \{x_{i,S}\}} 
$.
\item The endowment gap with respect to $A$ equals to  
$$
\sup \bigg\{ \frac{\sum_{i=1}^n \sum_{S \subseteq M} x_{i,S} \cdot v_i(S) - \psi_{A, \{x_{i,S}\}}}{\sum_{i=1}^n v_i(A_i) - \psi_{A, \{x_{i,S}\}}} \bigg|{\{x_{i,S}\}}\  s.t.\  {\sum_{i=1}^n v_i(A_i) - \psi_{A, \{x_{i,S}\}}}>0\bigg\}
$$
\end{enumerate}
\end{claim}
\begin{proof}
%
%
By Corollary \ref{cor-integralitygap-is-1}, $A$ can be $\alpha$-supported if and only if for every fractional solution $\{x_{i,S}\}$:
\begin{align*}
 \sum_{i=1}^n v_i^{A_i, \alpha}(A_i) \geq \sum_{i=1}^n \sum_{S \subseteq M} x_{i,S} \cdot v_i^{A_i,\alpha}(S)
\end{align*}


Additionally, since for every bundle $S$, $v_i^{A_i,\alpha}(S) =  v_i(S) +(\alpha-1) \cdot v_i(S \cap A_i)$, for every fractional solution $\{x_{i,S}\}$ it holds that:
  \begin{align*}
  \sum_{i=1}^n \sum_{S \subseteq M} x_{i,S} \cdot v_i^{A_i,\alpha}(S) = \sum_{i=1}^n \sum_{S \subseteq M} x_{i,S} \cdot   v_i(S) +(\alpha-1) \underbrace{\sum_{i=1}^n \sum_{S \subseteq M} x_{i,S} \cdot v_i(S \cap A_i)}_{\psi_{A, \{x_{i,S}\}}}
  \end{align*}

%
%
Note that the combination of the above two facts already establishes that if $A$ is $\alpha$-supported then claim (\ref{claim-psi-first-part}) holds.

We now continue to prove the second part. Rearranging, we have that $A$ is $\alpha$-supported if and only if for every fractional solution: 
\begin{align}\label{eq-main-end-vs-integ}
  \alpha  \cdot \sum_{i=1}^n v_i(A_i) \geq  \alpha \cdot \psi_{A, \{x_{i,S}\}} + \sum_{i=1}^n \sum_{S \subseteq M} x_{i,S} \cdot v_i(S) - \psi_{A, \{x_{i,S}\}} 
\end{align}
Consider the expression $\lambda_{A, \{x_{i,S}\}}=\alpha  \cdot \sum_{i=1}^n v_i(A_i) -  \alpha \cdot \psi_{A, \{x_{i,S}\}}$. Observe that $\lambda_{A, \{x_{i,S}\}}\geq 0$, simply because $\psi_{A, \{x_{i,S}\}}$ is composed of a sum of linear combinations of subsets of $A_i$, for each player $i$ (and the valuations are monotone). 

If $\lambda_{A, \{x_{i,S}\}}>0$ then obviously there is a large enough value of $\alpha$ such that inequality (\ref{eq-main-end-vs-integ}) holds.
Suppose that $\lambda_{A, \{x_{i,S}\}}=0$. Observe that $\sum_{i=1}^n \sum_{S \subseteq M} x_{i,S} \cdot v_i(S) \geq \psi_{A, \{x_{i,S}\}} $, simply because each term $v_i(S)$ in the LHS is replaced by $v_i(S\cap A_i)$ in the RHS and the valuations are monotone. 

If $\sum_{i=1}^n \sum_{S \subseteq M} x_{i,S} \cdot v_i(S) > \psi_{A, \{x_{i,S}\}} $ and $\lambda_{A, \{x_{i,S}\}}=0$,  which holds for any $A$ that is not maximal,
then no value of $\alpha$ makes inequality (\ref{eq-main-end-vs-integ}) hold and thus this allocation cannot be supported by any $\alpha$. However, if $\sum_{i=1}^n \sum_{S \subseteq M} x_{i,S} \cdot 
v_i(S) = \psi_{A, \{x_{i,S}\}} $ then any value of $\alpha$ makes the inequality hold.

We have thus identified that for an allocation not to be supported by any $\alpha$ it must be that there is some fractional solution $\{x_{i,S}\}$ for which $\lambda_{A, \{x_{i,S}\}}=0$ and $\sum_{i=1}^n \sum_{S \subseteq M} x_{i,S} \cdot v_i(S) > \psi_{A, \{x_{i,S}\}} $. If this is not the case then by rearranging inequality (\ref{eq-main-end-vs-integ}) we can determine the minimal value of $\alpha$ that supports $A$:
$$
\sup \bigg\{ \frac{\sum_{i=1}^n \sum_{S \subseteq M} x_{i,S} \cdot v_i(S) - \psi_{A, \{x_{i,S}\}}}{\sum_{i=1}^n v_i(A_i) - \psi_{A, \{x_{i,S}\}}} \bigg|{\{x_{i,S}\}}\  s.t.\  {\sum_{i=1}^n v_i(A_i) - \psi_{A, \{x_{i,S}\}}}>0\bigg\}
$$
Note that the supremum is bounded, since we assume that $A$ can be supported by some $\alpha$.
%
%
%
%
%
\end{proof}
%

Claim \ref{claim-endowment-vs-integarlity} implies that the endowment gap is at least the integrality gap: take $A$ to be any allocation and $\{x_{i,S}\}$ to be a fractional welfare maximizing solution:
  \begin{align*}
\alpha \geq  \frac{\sum_{i=1}^n \sum_{S \subseteq M} x_{i,S} \cdot v_i(S) - \psi_{A, \{x_{i,S}\}}}{\sum_{i=1}^n v_i(A_i) - \psi_{A, \{x_{i,S}\}}}\geq   \frac{\sum_{i=1}^n \sum_{S \subseteq M}x_{i,S} \cdot v_i(S)}{\sum_{i=1}^n v_i(A_i)} 
\end{align*}
and the right hand side is obviously at least the integrality gap.

However, as we will see in the paper, the integrality gap is usually strictly larger than the endowment gap. We now show that this is generically true for every instance with subadditive valuations with an integrality gap bigger than $1$. \shortversion{The proof can be found in the full version.}
\begin{claim}\label{claim-subadditive}
Consider an instance with two subadditive valuations. Suppose that the integrality gap of this instance is $y>1$. Then, for every small enough $\delta>0$ there is an instance with integrality gap $x=y\cdot \frac{(1+\delta)}{(1+\delta y)}$ in which the endowment gap is {strictly bigger than $x$}.
\end{claim}
\fullversion{
\begin{proof}
Let $(v'_1,v'_2)$ be two subadditive valuations. Denote by $OPT'$ the welfare of an optimal integral solution and by $OPT'^*$ the welfare of an optimal fractional solution $\{x_{i,S}\}$ with respect to $(v'_1,v'_2)$ (so $\frac {OPT'^*} {OPT'}=y$). For each bidder $i$, consider the valuation $v_i(S)=v'_i(S)+|S|\cdot \epsilon$, where $\epsilon=\frac {\delta\cdot OPT'^*} {m}$. Note that $v_i$ is still a subadditive function. For the instance $(v_1,v_2)$, let $OPT$ be the welfare of the optimal integral solution and $OPT^*$ be the welfare of the optimal fractional solution.
Observe that the welfare of any allocation $S$ with respect to $(v_1,v_2)$ is larger than the welfare of that allocation with respect to $(v'_1,v'_2)$ by exactly $\epsilon$ times the number of allocated items in $S$.
Thus, an optimal allocation (fractional or integral) in the instance $(v'_1,v'_2)$ in which all items are allocated is also optimal for $(v_1,v_2)$ and the difference in the welfare is exactly $m\cdot \epsilon$. Therefore, $OPT = OPT' + m \cdot \epsilon$ and $OPT^* = OPT'^* + m \cdot \epsilon$. Let $x$ denote the integrality gap of the instance $(v_1, v_2)$, 
we have that
\begin{align*}
x=\frac {OPT^*} {OPT} = \frac{OPT'^* + m \cdot \epsilon}{OPT' + m \cdot \epsilon} = \frac{(1+\delta)OPT'^*}{OPT'+\delta OPT'^* } = \frac{(1+\delta)OPT'^*}{(1+\delta y) OPT' } = y\frac{(1+\delta)}{(1+\delta y)}
\end{align*}

Let $(A_1,A_2)$ be some allocation that can be supported by $\alpha$ in the instance $(v_1,v_2)$. We claim that $A_1\cup A_2=M$. Else, there is some item $j$ that is not allocated and thus its price is $0$. Observe that given any bundle, item $j$ has a positive marginal value of at least $\epsilon=\frac {\delta\cdot OPT'^*} {m}>0$ for player $1$. Therefore, the bundle $A_1\cup\{j\}$ has a strictly larger profit than his equilibrium allocation $A_1$, a contradiction.

We will show that in the instance $(v_1,v_2)$, for any integral solution $A=(A_1,A_2)$ such that $A_1\cup A_2=M$, $\psi_{A, \{x_{i,S}\}}\geq (x-1)\cdot OPT$. We can then apply Claim \ref{claim-endowment-vs-integarlity} which says that the endowment gap is at least: 
$$
\frac {OPT^*-\psi_{A, \{x_{i,S}\}}} {OPT-\psi_{A, \{x_{i,S}\}}}\geq \frac {x\cdot OPT-(x-1)\cdot OPT} {OPT-(x-1)\cdot OPT}=\frac {1} {2-x} $$
{This completes the proof since $\frac {1} {2-x}>x =y\cdot \frac{(1+\delta)}{(1+\delta y)}$, where we use the fact that for any instance with two subaddititve players the integrality gap is strictly smaller than $2$ (see Appendix \ref{app-subaddititve}).}

We next show that in the instance $(v_1,v_2)$, for any integral solution $A=(A_1,A_2)$ such that $A_1\cup A_2=M$, $\psi_{A, \{x_{i,S}\}}\geq (x-1)\cdot OPT$. 
Observe that by subadditivity $v_i(S\cap A_i) \geq v_i(S)-v_i(S-A_i)$, thus:
\begin{align*}
\psi_{A, \{x_{i,S}\}} =\sum_{i=1}^2 \sum_{S \subseteq M} x_{i,S} \cdot v_i(S \cap A_i)
\geq  \underbrace{\sum_{i=1}^2\sum_{S\subseteq M}x_{i,S}v_i(S)}_{OPT^*=x\cdot OPT} - \sum_{i=1}^2 \sum_{S \subseteq M} x_{i,S} \cdot v_i(S - A_i)
\end{align*}

To complete the proof, we show that $\sum_{i=1}^2 \sum_{S \subseteq M} x_{i,S} \cdot v_i(S - A_i) \leq OPT$.
Observe that since $A_1\cup A_2=M$, in $\sum_{i=1}^2 \sum_{S \subseteq M} x_{i,S} \cdot v_i(S - A_i)$ we only assign player $1$ subsets of $A_2$ and player $2$ subsets of $A_1$. Taking into account that for each player $i$, $\sum_Sx_{i,S}\leq 1$, we get that $\sum_{i=1}^2 \sum_{S \subseteq M} x_{i,S} \cdot v_i(S - A_i)\leq v_1(A_2)+v_2(A_1) \leq OPT$.
\end{proof}
}

The claim provides a generic way of proving lower bounds on the endowment gap of subclasses of subadditive valuations: start with an instance with the maximal integrality gap in the subclass. The claim guarantees that there is an instance with an endowment gap that is strictly higher than the maximal integrality gap. A more careful look at the proof shows that the new instance belongs to the subclass as long as the subclass is closed under sum, like submodular and XOS valuations (a class of valuations $\mathcal V$ is closed under sum if for each $v,u\in\mathcal V$ we also have that $v+u\in \mathcal V$). We note that although the claim guarantees a generic method of proving lower bounds on the endowment gap, in the specific settings we study in this paper we are able to beat these bounds by introducing specific instances with stronger guarantees.

\section{The Main Result: The Endowment Gap of Submodular Valuations is at Most $2$}
\label{sec:submodular}


ֿIn this section we prove our main positive result: the endowment gap for submodular valuations is at most $2$. We prove this by showing that any allocation that is a ''local optimum'' of the social welfare function can be supported for $\alpha=2$ with prices that are equal to the marginal value of the items for the player that receives each item. We start by defining the notion of local optimum.  

\begin{definition}
An allocation $(O_1,\ldots, O_n)$ is a \emph{local optimum} if $\cup_{i=1}^n O_i=M$, and for every pair of players $i$ and $i'$ and item $j\in O_i$ we have that
$v_i(O_i) + v_{i'}(O_{i'}) \geq v_i(O_i-\{j\})+v_{i'}(O_{i'}\cup \{j\}).$

\end{definition}
In other words, in a local optimum {every item is allocated to some player},  and reallocating any single item does not improve the welfare. Note that any welfare maximizing allocation is in particular a local optimum. 
We are now ready to state our main positive result. 
\begin{theorem}
Let $v_1,\ldots, v_n$ be submodular valuations. Let $O=(O_1,\ldots, O_n)$ be a local optimum. Then $O$ is supported by any $\alpha\geq 2$. As an immediate corollary, the endowment gap of every instance with submodular valuations is at most $2$.
\end{theorem}
\begin{proof}
We explicitly construct prices that show that $O$ is supported by $2$. For each item $j \in O_i$ we define its price to be $p_j=v_i(j|O_i-j)$. Using the following two claims we show that for $\alpha\geq 2$ the prices $(p_1,\ldots, p_m)$  and the allocation $(O_1,\ldots,O_n)$ form an $\alpha$-endowed equilibrium. Later we will observe that our proofs hold even for lower prices.

We start by showing that with these prices and $\alpha\geq 2$, no player can gain by discarding items from his endowment.
\begin{claim} \label{claim-alg-keep}
Consider player $i$ that is allocated bundle $S_i$. Suppose that the price of each item $j\in S_i$ is $p_j= v_i(j|S_i-\{j\})$. Then, if $\alpha\geq 2$ the profit of player $i$ from every bundle $S'$ is at most the profit of $S_i\cup S'$. I.e., $v_i^{S_i,\alpha}(S')-\sum_{j\in S'} p_j\leq v_i^{S_i,\alpha}(S_i \cup S')-\sum_{j\in S_i \cup S'} p_j$.
\end{claim}
\begin{proof}
	We compare the profit of player $i$ from bundle $S'$:
		\begin{align*}
	v_i^{S_i,\alpha}(S')-\sum_{j\in S'} p_j = v_i(S' ) + (\alpha-1) \cdot v_i(S_i\cap S')  -  \sum_{j \in S'} p_j
	\end{align*}
	to his profit from the bundle $S' \cup S_i$:
	\begin{align*}
	v_i^{S_i,\alpha}(S_i \cup S')-\sum_{j\in S' \cup S_i} p_j =v_i(S' \cup S_i )+(\alpha-1) \cdot v_i(S_i) -  \sum_{j \in S' \cup S_i} p_j 
	\end{align*}
	Using the fact that $v_i(S_i) = v_i(S_i- S' |S_i\cap S') +   v_i(S' \cup S_i )$ and rearranging the last expression, we get that the profit of bundle $S' \cup S_i$ equals:
	\begin{align*}
 \underbrace{v_i(S' \cup S_i )+(\alpha-1) \cdot v_i(S_i\cap S')}_{\geq v_i^{S_i,\alpha}(S') \text{~by monotonicity}}  -  \sum_{j \in S'} p_j+(\alpha-1) \cdot v_i(S_i- S' |S_i\cap S') - \sum_{j \in S_i -S'} p_j.
	\end{align*}
Thus, in order to show that $v_i^{S_i,\alpha}(S_i \cup S')-\sum_{j\in S' \cup S_i} p_j \geq v_i^{S_i,\alpha}(S')-\sum_{j\in S'} p_j$, it suffices to show that $(\alpha-1) \cdot v_i(S_i- S' |S_i\cap S') - \sum_{j \in S_i -S'} p_j\geq 0$. 
Since $\alpha\geq 2$, it holds that $\alpha-1\geq 1$ and so to show this it is enough to prove that $ v_i(S_i- S' |S_i\cap S') \geq  \sum_{j \in S_i -S'} p_j$ for every submodular valuation.
Towards this end, denote the items in $S_i -S'$ by ${1,\ldots,|S_i-S'|}$. With this notation we have that 
	$v_i(S_i- S' |S_i\cap S') = \sum_{j=1}^{|S_i-S'|} v_i(j |(S_i\cap S')\cup\{1,\ldots j-1\})$. Finally, observe that by submodularity we have that for every $1 \leq j \leq |S_i-S'|$ it holds that $p_j = v_i(j|S_i -j) \leq v_i(j|S_i \cup S' -\{j, j+1, \ldots ,|S_i-S'|\} ) = 
	v_i(j |(S_i\cap S')\cup\{1,\ldots j-1\})$.
\end{proof}


We next show that with these prices, an agent can never gain by adding items to his endowment. 

\begin{claim}\label{claim-alg-no-add}
Let $O=(O_1,\ldots,O_n)$ be a local optimum and suppose that the price of each item $j\in O_i$ is $p_j = v_i(j|O_i-j)$. Then, for any $\alpha \geq 1$ and bundle $T$ the profit from $O_i$ is at least the profit from $O_i\cup T$. 
\end{claim}
\begin{proof}
Assume without loss of generality that $O_i\cap T = \emptyset$. Observe that since the allocation $O$ is a local optimum we have that for any $j\in T$, $v_i(j|O_i) \leq p_j =v_k(j|O_k-j)$ for player $k$ such that $j\in O_k$. Denote the items in $T$ by $1,\ldots,|T|$. Since the valuations are submodular we have that:
		\begin{align*}
		\alpha \cdot v_i(O_i)+v_i(T-O_i|O_i) - \sum_{j \in O_i\cup T} p_j
		 &=\alpha \cdot v_i(O_i) - \sum_{j \in O_i} p_j + \sum_{j=1}^{|T|} v_i(j|O_i\cup\{1,\ldots,j-1\}) - \sum_{j \in T} p_j  \\
		 &\leq\alpha \cdot v_i(O_i) - \sum_{j \in O_i} p_j + \sum_{j=1}^{|T|} v_i(j|O_i) - \sum_{j \in T} p_j\\
		 &\leq 
		  \alpha\cdot  v_i(O_i) - \sum_{j \in O_i} p_j
		\end{align*}
where in the second-to-last inequality we use the submodularity of $v_i$ to claim that $v_i(j|O_i)\geq v_i(j|O_i\cup\{1,\ldots,j-1\})$.
\end{proof}

We apply claims \ref{claim-alg-keep}  and \ref{claim-alg-no-add} to conclude the proof of the theorem. 
Let $O=(O_1,\ldots,O_n)$ be a local optimum, recall that in a local optimum we have that $\cup_{i=1}^n O_i=M$, and suppose that the price of each item $j\in O_i$ is $p_j = v_i(j|O_i-j)$. To complete the proof we show that any player $i$ demands the set $O_i$ at the these prices. 
By Claim \ref{claim-alg-no-add}, for $\alpha \geq 1$ 
for any bundle $T$:
$$ v^{O_i,\alpha}_i(O_i)-\sum_{j\in O_i }p_{j}\geq v^{O_i,\alpha}_i(O_i\cup T)-\sum_{j\in O_i \cup T}p_{j} $$
and by Claim \ref{claim-alg-keep} for $\alpha\geq 2$ and any bundle $T$:
$$ v_i^{O_i,\alpha}(O_i\cup T)-\sum_{j\in O_i \cup T}p_{j}\geq v_i^{O_i,\alpha}(T)-\sum_{j\in T}p_{j} $$
Combining the two inequalities we get that for any $\alpha\geq 2$ and any bundle $T$: 
$$ v^{O_i,\alpha}_i(O_i)-\sum_{j\in O_i }p_{j}\geq v_i^{O_i,\alpha}(T)-\sum_{j\in T}p_{j} $$
as needed. 
\end{proof}

We note that the proof still holds even if we reduce the prices such that the price of item $j \in O_i$ is 
$\max_{i'\neq i } v_{i'}(j|O_{i'})\leq p_j\leq v_i(j|O_i-\{j\}) $. 
That is, the price $p_j$ can be reduced to the second highest marginal value for the item. The reason is simple: the profit of player $i$ from the bundle $O_i$ has increased at least as any other bundle in this reduction, so Claim \ref{claim-alg-keep} still holds. Claim \ref{claim-alg-no-add} also holds, as for the proof to hold we only need the price of each item to be the second highest marginal value.

\subsection{An Implication: The Approximation Ratio of a Local Maximum}\label{subsec-implication}
One interesting corollary of the algorithm is not directly related to the endowment gap. In previous work it was implicit that the welfare of any local optimum is at least half of the value of the welfare maximizing \emph{integral} solution. By a direct application of Corollary \ref{cor-fractional-approximation}, we are able to strengthen this result and show that the welfare of any local optimum is at least half of the welfare of the welfare maximizing \emph{fractional} solution:
\begin{proposition}
Let $(O_1,\ldots, O_n)$ be a local optimum. Let $\{x_{i,S}\}$ be some fractional solution for the LP presented in Section \ref{sec:int-gap}. Then:
$$
2\cdot \sum_{i=1}^n v_i(O_i)\geq \sum_{i=1}^n \sum_{S \subseteq M}x_{i,S} \cdot v_i(S)
$$
\end{proposition}	

\subsection{Tightness of Analysis}

The following proposition shows that considering
local optima (or even global ones) will not help us prove a better bound than $2$ on the endowment gap for submodular valuations:
\begin{proposition}
For any $\delta>0$, there is an instance with submodular valuations in which every local optimum maximizes the welfare and these local optima cannot be supported by $\alpha<2-\delta$. In this instance there is another allocation that can be supported by $\alpha=1.5+\delta$.
\end{proposition}
\begin{proof}
Consider an instance with four submodular players ($a_1, a_2, b_1, b_2$) and $2k+1$ items ($k\geq 2$). Each of the items belongs to one of the following three sets: $X,Y, \{c\}$, 
where $X=\{x_1,...,x_k\}$ and $Y=\{y_1,...,y_k\}$. 
The valuations of the players are as follows: 
\begin{itemize}
	\item $a_1$ has a unit demand valuation with value $\frac 1 k$ for each of the items in $X$ (and $0$ for the rest).
	\item $a_2$ has a unit demand valuation with value $\frac 1 k$ for each of the items in $Y$ (and $0$ for the rest).
	\item   The valuations of $b_1,b_2$ are defined as follows: given a set $T$, let $v_T$ be the budget additive valuation with budget $1$ that gives value $\frac 1 k$ for every item in $T$ and value of $1$ for item $\{c\}$. Then the valuation of $b_1$ is $b_1(S)=v_X(S)+|S|\cdot \epsilon$ and $b_2(S)=v_Y(S)+|S|\cdot \epsilon$,
	where $\epsilon>0$ is small enough.  
\end{itemize}
Observe that the valuations of $b_1$ and $b_2$ are submodular as they are the sum of two submodular valuations.

We next show that up to symmetry, there is only one locally optimal allocation. 
In every local maximum item $c$ must be allocated to either $b_1$ or $b_2$ as its marginal value for both is positive given any bundle, and $a_1,a_2$ have a value of $0$ for item $c$. Without loss of generality assume it is allocated to $b_1$. Now, it is not possible that in a local optimum $b_1$ is allocated $X\cup \{c\}$ as his marginal value for each item in $X$ is only $\epsilon$, while $a_1$ values each item in $X$ at $1/k$. 
Thus, for $\epsilon<1/k$ player $a_1$ must receive an item from $X$, {without loss of generality} item $x_1$. Now, given any subset of $X-\{x_1\}$, $b_1$ has positive value for any additional item in $X-\{x_1\}$, and is the only {player} with such positive value, so he must get $X\cup \{c\}-\{x_1\}$. Finally, {all items in} $Y$ must be allocated to $b_2$ as the marginal value of any item in $Y$ (given any subset of $Y$) is larger for $b_2$ than for any other player, in particular $a_2$. 

We conclude that {in a local maximum, without loss of generality,} $a_1$ is allocated $\{x_1\}$, $b_1$ is allocated $X\cup \{c\}-\{x_1\}$ and $b_2$ is allocated $Y$.

Fix some value of $\alpha$ that supports this allocation. First, we observe that for any $j\in X-\{x_1\}$, $b_1(j|X\cup \{c\}-\{x_1,j\})=\epsilon$, thus $p_j\leq \alpha\cdot \epsilon$ (otherwise the profit of $b_1$ from the bundle $X\cup \{c\}-\{x_1\}-\{j\}$ is bigger than the profit of his equilibrium allocation). Similarly, $p_c\leq \alpha \cdot( \frac {1} k+\epsilon)$.

Then, it must be that:
\begin{enumerate}
\item $b_2$ prefers his allocation over item $c$: $(1+k\cdot \epsilon)\cdot \alpha-\sum_{y\in Y}p_{y}\geq 1+\epsilon-p_c$. 
\item $a_2$ that has zero profit in equilibrium has a non-positive profit from items in $Y$: for every $y\in Y$, $0\geq \frac 1 k-p_{y}$.
\end{enumerate}
Summing up all these inequalities with $p_c\leq \alpha \cdot( \frac {1} k+\epsilon) $, we get that: $ֿ\alpha\cdot (1+k\cdot \epsilon)\geq 2+\epsilon-\alpha\cdot (\frac 1 k+\epsilon)$. That is, $\alpha\geq \frac {2+\epsilon} {1+k\cdot \epsilon+(\frac 1 k +\epsilon)}$, which approaches $2$ for $k$ that approaches $\infty$ and a choice of $\epsilon=\frac 1 {k^2}$.

As for the second part of the proposition, we will now see that the allocation in which $a_1$ is allocated $\{x_1\}$, $b_1$ is allocated $X-\{x_1\}\cup \{c\}$, $b_2$ is allocated $Y-\{y_1\}$ and $a_2$ is allocated $\{y_1\}$ is supported by $1.5+\epsilon$. To see this, we simply note that with this value of $\alpha$ the following prices constitute an equilibrium with respect to this allocation: $p_c=\frac 1 k+\epsilon$, $p_x=\epsilon$ (for each $x\in X$), $p_y=\frac 1 {2k}+\epsilon$ (for each $y\in Y-\{y_1\}$), and $p_{y_1}=\frac 1 {k}+\epsilon$.
\end{proof}


\subsection{The Complexity of Finding a Local Optimum}
Our proof that the endowment gap for submodular valuations is at most $2$ starts with a local maximum (local optimum) $(O_1,\ldots, O_n)$. 
Can such a local optimum be efficiently found? We next show that there are both communication and computation hurdles in finding a local optimum:
We first show that finding a local optimum in combinatorial auctions with submodular valuations requires exponential number of value queries. In addition, we show that there exists some family of succinctly represented submodular valuations for which finding a local optimum is PLS hard.\footnote{We will not give a precise definition here (see \cite{johnson1988easy} for a formal definition), but the PLS class intuitively captures problems which admit a local search algorithm. An example for a PLS complete problem is finding a pure Nash equilibrium in congestions games \cite{fabrikant2004complexity}.} Both results hold even if there are only two players.

In light of the hardness of finding exact local optimum, one might be tempted to use an approximate local optimum instead, as it is computationally feasible to find, and hope it can be supported by a small $\alpha$. However, it is easy to provide examples in which a $(1-\epsilon)$-approximate local optimum cannot be supported by any $\alpha$. Consider a setting with two additive bidders and two items. The first additive bidder has a value of $1-\epsilon$ for item $a$, and $\epsilon$ for item $b$. The second bidder has a value of $0$ for both items. Observe that allocating $a$ to the first bidder and $b$ to the second one is an $(1-\epsilon)$-local optimum. However, this allocation cannot be supported by any $\alpha$: for the bundle $\{b\}$ to be in the demand of the second player, the price of item $b$ must be $0$. However, the profit of the bundle $\{ab\}$ for player $1$ is then strictly bigger than that of $\{a\}$.
We thus leave as an open question the problem of efficiently computing an $\alpha$-endowed equilibrium for combinatorial auctions with submodular bidders for a small value of $\alpha$ (\fullversion{Proposition \ref{prop-exists-alpha} shows}\shortversion{in the full version we show} that it is possible to efficiently compute an allocation that is supported by some huge $\alpha\geq 1$).

We now turn to proving the hardness results for finding a local optimum\shortversion{ (proof in the full version)}:
\begin{proposition}\label{claim-value-hardness}
Finding a local optimum in a combinatorial auction with two submodular valuations requires exponentially many value queries, even with randomized algorithms.
\end{proposition}
\fullversion{
\begin{proof}
	Let the number of items be $m=2k+1$ for some integer $k\geq 1$. 
	In the proof the valuation of each player $i$ will belong to the following family of valuations parametrized by $c{^i_S}$ satisfying $0\leq c^i_S< \frac 1 2$ for every set $S$ {with} $|S|=k+1$:
	
\begin{align*} 
	v_i(S)=\left\{
	\begin{array}{ll}
	|S|, & |S|\leq k\\
	k+\frac 1 2+c{^i_S}, & |S|=k+1 \\
	k+1, & |S|\geq k+2.
	\end{array}
	\right.
\end{align*}
Notice that $v_i$ is a monotone submodular valuation.
	
	Given two submodular valuations  $v_1$ and $v_2$ from the family above, what can we say about their local optimum? Consider an allocation $(S,M-S)$. If the allocation is a local optimum, then clearly $|S|=k$ or $|M-S|=k$ as otherwise reallocating a single item from the larger bundle decreases the value of the large bundle by at most $\frac 1 2 -c_S$ and increases the value of the small bundle by $1$, so the overall welfare has strictly increased.
	Thus, only allocations $(S,M-S)$ where one of the bundles has size $k$ can potentially be a local optimum. 
	
Recall that a local maximum of a labeled graph is a vertex whose label is at least as large as the labels of its neighbors. For our reduction we use odd graphs. An odd graph is constructed\footnote{The textbook definition of an odd graph starts with a set $M$ of size $2k+1$ and associates each of the vertices with 
	a set $S \subseteq M$ of size $k$. 
	Two vertices $v_S$ and $v_{S'}$ are connected if and only if $S\cap S'=\emptyset$. The definition we give above is easier to work with. Note that that this definition and the one we use above are equivalent as can be seen by changing the label from $S$ in the textbook definition to $M-S$.} by starting with a set $M$ of size $2k+1$ and associating each of the vertices with 
a set $S \subseteq M$ of size $k+1$. Two vertices $v_S$ and $v_{S'}$ are connected if and only if $|S\cap S'|=1$.

We will show that finding a local optimum when the valuations belong to the family defined above is equivalent to finding a local maximum of an odd graph. This is enough to complete the proof as the results of Santha and Szegedy \cite{santha2004quantum} imply that the number of queries needed to find a local optimum in an odd graph is $exp(m)$. This result is also holds for randomized algorithms (and in fact also for quantum algorithms).

Given a labeled odd graph with $M$ of size $2k+1$, we reduce it to combinatorial auctions with submodular valuations as follows.	
Assume without loss of generality that the labels on the odd graph are less than $\frac{1}{2}$ (this can be achieved, e.g., by dividing all numbers by a large enough number). We define identical valuations for the two players to be the above parametrized valuations where the parameter $c_S=c^1_S=c^2_S$ is equal to the label of the unique vertex in the odd graph that is associated with $S$, for each bundle $S$ with $|S|=k+1$. We next show that local maxima in this combinatorial auction correspond to local maxima in the odd graph.

Consider an allocation $(S,M-S)$ with $|S|=k+1$. The welfare of this allocation is $v_1(S)+v_2(M-S)=k+\frac 1 2+c{^i_S}+k=2k+\frac 1 2  {+}  c_S$. Similarly, if $|M-S|=k+1$ then the welfare is $2k+\frac 1 2 {+} c_{M-S}$. Recall that
an allocation is a local maximum if and only if moving item $j$ to the other player does not improve the welfare. Thus, $2k+\frac 1 2 {+}  c_{S}\geq 2k+\frac 1 2  {+} c_{(M-S)\cup \{j\}}$ for every $j\in S$. In other words, $c_{S}\geq  c_{(M-S)\cup \{j\}}$ for every $j\in S$. Recall that in the odd graph the vertex that corresponds to $S$ is connected exactly to all vertices which correspond to $(M-S)\cup\{j\}$ for every $j\in S$. This immediately implies that an allocation $(S,M-S)$ ($|S|=k+1$) in the combinatorial auction is a local maximum if and only if the vertex that corresponds to $S$ is a local maximum in the odd graph. The proof for allocations $(S,M-S)$ where $|M-S|=k+1$ is symmetric.
\end{proof}
} 

The proposition proves that finding a local maximum requires exponentially many value queries. The other common model for accessing the valuations is the more general communication model. That is, how many bits do the players need to exchange in order to compute a local optimum? One obstacle in proving hardness in the communication model is that the proof of Proposition \ref{claim-value-hardness} is based on the result of \cite{santha2004quantum} which proves the hardness of finding a local maximum on the odd graph using value queries. An analogous result for the communication model was not known when the first version of the paper was written.

However, inspired by our work, the paper \cite{BDN18} studies a communication variant of local search on the odd graph. Using this result and a very similar reduction to that of Proposition \ref{claim-value-hardness}, we show:

\begin{proposition}\label{prop-comm-hardness}
The communication complexity of finding a local optimum in a combinatorial auction with two submodular valuations is $exp(m)$. This results holds even for randomized protocols.
\end{proposition}

We give a sketch of the proof in \fullversion{Appendix \ref{appendix-comm-hardness}}\shortversion{the full version}.

\begin{proposition}
There exists a family of succinctly described submodular functions for which computing a value query can be done in polynomial time but finding a local optimum in a combinatorial auction with two valuations that belong to this family is PLS-hard.
\end{proposition}

\begin{proof}
We reduce from the PLS complete problem of finding a locally optimal cut in a graph $G$. In this problem, we are given an edge-weighted graph $G=(V,E)$, and we are looking for a partition of the vertices $(S,V-S)$ such that the (weighted) cut cannot be improved by moving any single vertex from one side of the cut to the other.

We now reduce this problem to a combinatorial auction with two identical submodular valuations. The items in the combinatorial auction are the vertices of the graph. For each set of vertices $S$, we set $v(S)$ to be the sum of the weights of all edges that touch at least one vertex in $S$. It is not hard to see that this valuation is monotone and submodular.

Consider an allocation $(S,M-S)$. Observe that the welfare of $(S,M-S)$ is exactly the sum of the weights of all edges in the graph plus the sum of edges in the cut: the weight of an edge $e=(v,u)$ is counted once if and only if both are in $S$ or if both are in $M-S$. It is counted twice if and only if $v\in S$ or $u\in M-S$ or  $u\in S$ or $v\in M-S$. This implies that $(S,M-S)$ is a local optimum in the combinatorial auction if and only if $(S,V-S)$ is a locally optimal cut in the graph, as needed.
\end{proof}

\section{Lower Bounds on the Endowment Gap}

We now prove some lower bounds on the endowment gap. In Section \ref{sec:submodular} we have shown that the endowment gap for submodular valuations is at most $2$. What about larger classes of valuations, such as XOS or subadditive valuations? In sharp contrast to the case of submodular valuations, we show that with the larger classes, the endowment gap cannot be bounded uniformly for the entire class, even if there are only two players. 
Note that this does not contradict our claim \fullversion{(Proposition \ref{prop-exists-alpha})} that for any instance, there exists an allocation that is supported by \emph{some} $\alpha$, as we now only show that for every \emph{fixed} $\alpha$ there is some instance such that no allocation is supported by that value of $\alpha$.

\begin{proposition}\label{prop-XOS}
Fix $\alpha>1$. There exists an instance that consists of only three identical items and two players with identical XOS valuations such that no allocation is supported by $\alpha$.
\end{proposition}
\begin{proof}	
Consider an instance with two bidders and three identical items $x_1,x_2,x_3$. The valuations of the bidders are identical (since the items are identical we slightly abuse notation and use $v(x)$ to specify the value of every bundle $S$ such that $|S|=x$): $v(1)=1, v(2)=1+\frac 1 {12\alpha^2}$, $v(3)=1+\frac 1 {3\alpha}$. We prove that this is indeed an XOS valuation by providing an explicit construction of the clauses of $v$: a clause $(x_j=1)$ for every item $x_j$, a clause $(x_j=\frac 1 2 +\frac 1 {24\alpha^2},x_{j'}=\frac 1 2 +\frac 1 {24\alpha^2})$, for every pair of $x_{j'}$ and $x_j$, $j'\neq j$, and a clause $(x_1=\frac 1 2, x_2=\frac 1 2, x_3=\frac 1 {3\alpha})$.

Note that in every equilibrium allocation all items must be allocated: if item $x_j$ is unallocated, then its price is $0$. However, since the valuations of the players are strictly increasing, the profit of his equilibrium allocation with $x_j$ is strictly bigger than his profit from his equilibrium allocation.


Thus, from now on we only consider allocations that allocate all items. There are two possible allocations that allocate all items (up to symmetry). We will show that both allocations are not supported by $\alpha$. In the first allocation, one player, say player $1$, gets all three items. The marginal value of an item for player $1$ is $\frac 1 {3\alpha}-\frac 1 {12\alpha^2}<\frac 1 {3\alpha}$. Thus the price of every item in equilibrium is at most $\alpha\cdot \frac 1 {3 \alpha}=\frac 1 3$ (if there is an item with a higher price, the profit for player $1$ of the bundle that contains the other two items is higher than the profit of the grand bundle). In this case the profit of player $2$ from taking one item is positive: $1-\frac 1 3=\frac 2 3$. A contradiction to the assumption that the empty bundle maximizes the profit of player $2$.

The other allocation is when one player, say player $1$, is allocated two items and player $2$ receives only one item. This allocation is not supported by $\alpha$ as well: since $v(1)=1 $ and $v(2)=1+  \frac 1 {12\alpha^2}$, the marginal value of player $1$ for any item is $\frac 1 {12\alpha^2}$. Thus, taking the endowment effect into account, the price of each item that is allocated to player $1$ cannot exceed $\frac 1 {12\alpha}$. But now the profit of the grand bundle for player $2$ is strictly higher than its current single item allocation: the marginal value of player 2 for the two items that were allocated to player $1$ is $\frac 1 {3\alpha}$, 
while the sum of prices of these items is at most $\frac 1 {6\alpha}$. We have reached a contradiction once again.
\end{proof}

For submodular valuations we can prove that the endowment gap is at least $\frac 3 2$. This leaves us with a small gap to the upper bound of $2$.

\begin{proposition}
There exists an instance of two players with submodular valuations where the endowment gap is at least $\frac 3 2$.
\end{proposition}
\begin{proof}
Feige and Vondrak \cite{feige2006approximation} show that the integrality gap of two players with submodular valuations is at least $\frac 6 5$. We now show that this example shows that the endowment gap is at least $\frac 3 2$.

As mentioned, we have two players, call them Alice and Bob. There are $4$ items, the value of each singleton is $1$ the value of each bundle of three or more items is $2$. The values of pairs of items are:
$$\begin{array}{ccc}
 & Alice & Bob \\ 
\{ab\} & 2  &  4/3\\ 
\{cd\} & 2  &  4/3\\ 
\{ac\} & 4/3 & 2\\ 
\{bd\} & 4/3 & 2 \\ 
\{ad\} & 5/3 & 5/3  \\
\{bc\} & 5/3 & 5/3
\end{array} $$
There are $4$ possible allocations of all items (up to symmetry), and we will see that none of them is supported by $\alpha<\frac 3 2$. We will use the fact that the value of the optimal fractional solution is $4$ (Alice receives $\frac 1 2$ of each of the sets $\{ab\}, \{cd\}$, Bob receives $\frac 1 2$ of each of the sets $\{ac\},\{bd\}$: $x_{A,\{ab\}}=x_{A,\{cd\}}=x_{B,\{ac\}}=x_{B,\{bd\}}=\frac 1 2$). 
\begin{enumerate}
\item $(\{abcd\},\emptyset)$: The value of this allocation is $2$, and since there is a fractional solution with value $4$, by Claim \ref{claim-endowment-vs-integarlity} this allocation requires $\alpha\geq 2$ to support it.
\item $(\{abc\},\{d\})$: Suppose that there is an equilibrium with some $\alpha>1$ and prices $p_a,p_b,p_c,p_d$. $v_A(c|\{ab\})=0$, hence $p_c=0$ (as if $p_c>0$ the profit of Alice from the bundle $\{ab\}$ is strictly bigger than that of $\{abc\}$.) However, this means that $v_B(c|\{d\})-p_c=\frac 1 3$, thus
$$
\alpha \cdot v_B(\{d\})-p_d< \alpha \cdot v_B(\{d\})-p_d+ v_B(c|\{d\})-p_c
$$

I.e., the profit of Bob from the bundle $\{cd\}$ is strictly bigger than that of $\{d\}$. A contradiction to the assumption that there is an equilibrium.
\item $(\{ab\},\{cd\})$: Denote this allocation by $A$ and note that its welfare is $10/3$. 
Observe that:
$$\psi_{A, \{x_{i,S}\}} =\sum_{i=1}^n \sum_{S \subseteq M} x_{i,S} \cdot v_i(S \cap A_i)=\frac 1 2v_A(\{ab\})+\frac 1 2v_B(\{c\})+\frac 1 2v_B(\{d\})=2
$$

Thus, by Claim \ref{claim-endowment-vs-integarlity}, to support $A$ we must have $\alpha\geq \frac {4-2} {\frac {10} 3-2}=\frac {3}{2}$.
\item $(\{ad\},\{bc\})$: Denote this allocation by $A$ and note that its welfare is $10/3$. Observe that 
$$\psi_{A, \{x_{i,S}\}} =\sum_{i=1}^n \sum_{S \subseteq M} x_{i,S} \cdot v_i(S \cap A_i)=\frac 1 2v_A(\{a\})+\frac 1 2v_A(\{d\})+\frac 1 2v_B(\{b\})+\frac 1 2v_B(\{c\})=2.
$$
Similarly to before, by Claim \ref{claim-endowment-vs-integarlity}, to support $A$ we must have $\alpha\geq \frac {4-2} {\frac {10} 3-2}=\frac 3 2$.
\end{enumerate}
Finally, we note that the optimal integral allocation $(\{ab\},\{cd\})$ is indeed supported by $\alpha= 3/2$ by using the prices $p_a=p_b=1$, $p_c=p_d=2/3$. 
\end{proof}

The above claim shows that if we have two players with general submodular valuations then the endowment gap is at least $\frac 3 2$. Next, we consider a more restricted class -- budget additive valuations -- and show that even in this much simpler class the endowment gap is essentially the same. However, we do need more players to show this.

\begin{proposition}
For every $\epsilon>0$, there exists an instance of four players with budget additive valuations in which the endowment gap is at least $\frac 3 {2+\epsilon/2}$. 
\end{proposition}
\begin{proof}
We consider a modification of the integrality gap example of Chakrabarty and Goel \cite{chakrabarty2010approximability}. We have $4$ players ($a_1,a_2,b_1,b_2)$ and $5$ items ($x_1, x_2,y_1,y_2,c$). $a_1,a_2$ have budget $1$, $b_1,b_2$ have budget $2+\epsilon$, for some arbitrarily small $\epsilon>0$.

To complete the definition of the valuations, we only need to specify the values of players for single items:
\begin{itemize}
	\item Players $a_1$ and $b_1$ have value $1$ for items $x_1,x_2$.
	\item Players $a_2$ and $b_2$ have value $1$ for items $y_1,y_2$.
	\item Players $b_1,b_2$ value $2$ for item $c$.
	\item The rest of the values are $0$.
\end{itemize}

To analyze the endowment gap we provide several lemmas that characterize the allocations that can be supported in an endowed equilibrium. We then show that each of these allocations can be supported only by $\alpha \geq \frac 3 {2+\epsilon/2}$.
\begin{lemma}
Without loss of generality, in an endowed equilibrium $b_1$ is assigned item $c$. 
\end{lemma}
\begin{proof}
We claim that in equilibrium item $c$ must be allocated to one of the players $b_1,b_2$. If $c$ is not allocated, then its price must be $0$. The same holds if $c$ is allocated to player $a_1$ or player $a_2$, the value of both players for $c$ is $0$. Now observe that each player $b_i$ has three items with positive value, and that $b_1(c|\{x_1,x_2\})=\epsilon>0$. Thus, when $p_c=0$, the profit of player $b_1$ from a bundle that contains $c$ is always strictly bigger than his profit from a bundle without it. The argument for player $b_2$ is identical, and thus without loss of generality in equilibrium $b_1$ is assigned item $c$.
\end{proof}

\begin{lemma}
Without loss of generality, in an endowed equilibrium player $b_2$ is assigned at least one of $y_1,y_2$.
\end{lemma}
\begin{proof}
Observe that if player $b_1$ was not assigned any of the items $y_1,y_2$ , then his profit is $0$, since he has positive value only for items $c,y_1,y_2$ and did not get any of them. The only other player with positive value for $y_1$ and $y_2$ is $a_2$. However, if $a_2$ receives the bundle $\{y_1,y_2\}$ the marginal contribution of any of these items is $0$, thus the price of both of these two items is $0$. In that case, the profit of player $b_2$ from the bundle $\{y_1\}$ is strictly positive and is bigger than his $0$ profit in equilibrium, which is a contradiction. 
\end{proof}

\begin{lemma}
Without loss of generality, in an endowed equilibrium $a_1$ is assigned item $x_2$. 
\end{lemma}
\begin{proof}
We can claim that player $a_1$ is allocated at least one of $x_1,x_2$: the only other player with positive value for these items is $b_1$. Now, $b_1$ is already assigned item $c$, so the marginal contribution of the item is at most $\epsilon$ and thus the price is at most $\alpha\cdot \epsilon$. For $\epsilon< \frac 1 \alpha$, $\alpha\cdot \epsilon<1$ and the bundle $\{x_1\}$ has a positive profit for player $a_1$. This profit is bigger than his $0$ profit in equilibrium.
\end{proof}

\begin{lemma}
Without loss of generality, in equilibrium player $b_1$ is assigned $x_1$ and $p_{x_1}<\alpha\cdot \epsilon$.
\end{lemma}
\begin{proof}
Observe that $x_1$ can only contribute positively to players $a_1,b_1$ and that $a_1(x_1|\{x_2\})=0$. Thus, if $x_1$ is allocated to $a_1$ then its price must be $0$, but then 
$b_1(x_1|\{c\})>0$, which implies that $b_1$ profit increases when he adds $x_1$ to his bundle. Similarly to before, this means that this is not an equilibrium allocation. Also, similarly to the previous lemma, $p_{x_1}<\alpha\cdot \epsilon$.
\end{proof}

This leaves us with the following two allocations that can be supported (up to symmetry):
\begin{enumerate}
\item $b_1$ gets $\{c,x_1\}$, $b_2$ gets $\{y_1,y_2\}$, $a_1$ gets $x_2$: in this case we have:
\begin{enumerate}
\item $b_2$ prefers his allocation over item $c$: $2\cdot \alpha-p_{y_1}-p_{y_2}\geq 2-p_c$.
\item $a_1$ prefers his allocation over item $x_1$ (recall that $p_{x_1}\leq \alpha\cdot \epsilon$): $\alpha-p_{x_2}\geq 1-\epsilon\cdot \alpha$.
\item $b_1$ prefers his allocation over $\{x_1,x_2\}$: $2\cdot \alpha -p_c\geq 1+\alpha -p_{x_2}$ (by rearranging $\alpha-p_c \geq 1 -p_{x_2}$). 
\item $a_2$ that has zero profit in equilibrium has a non-positive profit from $y_1$ and $y_2$: $0\geq 1-p_{y_1}$, $0\geq 1-p_{y_2}$.
\end{enumerate}

Summing these inequalities we get that $ \alpha\cdot (4+\epsilon)\geq 6$. I.e., to support this allocation we need $\alpha \geq \frac {3} {2+\epsilon/2}$. As $\epsilon$ approaches $0$ this ratio approaches $\frac 3 2$.

\item $b_1$ gets $\{c,x_1\}$, $b_2$ gets $y_1$, $a_1$ gets $x_2$, $a_2$ gets $y_2$: 
\begin{enumerate}
\item $b_2$ prefers his allocation over item $c$: $\alpha-p_{y_1} \geq 2-p_c$.
\item $b_2$ prefers his allocation over $\{y_1,y_2\}$: $\alpha-p_{y_1}\geq \alpha+1-p_{y_1}-p_{y_2}$. Hence $p_{y_2}\geq 1$.
\item $a_1$ prefers his allocation over item $x_1$ (recall that $p_{x_1}\leq \alpha\cdot \epsilon$): $\alpha-p_{x_2}\geq 1-\epsilon\cdot \alpha$.
\item $b_1$ prefers his allocation over $\{x_1,x_2\}$: $2\cdot \alpha -p_c\geq 1+\alpha -p_{x_2}$.
\item $a_2$ prefers his allocation over item $y_1$: $\alpha-p_{y_2}\geq 1-p_{y_1}$.
\end{enumerate}
Summing these inequalities we get that $ \alpha\cdot (4+\epsilon)\geq 6$. I.e., to support this allocation we need $\alpha \geq \frac {3} {2+\epsilon/2}$. As $\epsilon$ approaches $0$ this ratio approaches $\frac 3 2$.
\end{enumerate}

%
\end{proof}

\subsubsection*{Acknowledgments}

The second author was partially supported by BSF grant no. 2016192.
\bibliographystyle{plain}
\bibliography{n}

\fullversion{
\appendix

\section{Existence of Endowed Equilibrium}
\label{sec:support-instance}


We now show that for every instance there is some allocation that can be $\alpha$-supported with some $\alpha>1$. The caveat is that this $\alpha$ might be huge and instance dependent. We will bring two (similar) proofs for this: one that shows that there is a welfare maximizing allocation that can be supported, and another proof that shows how to find such allocation in a computationally efficient way.

The value of $\alpha$ used in Proposition \ref{prop-exists-alpha} is instance dependent and can be very large (depends on the value of $OPT$, the maximal welfare for the instance) and thus does not provide a uniform upper bound that hold for all instances. We show (Proposition \ref{prop-XOS}) that an upper bound that holds for all instances does not exist. 

\begin{proposition}\label{prop-exists-alpha}\
\begin{enumerate}
	\item In every instance there is some welfare maximizing allocation $O=(O_1,\ldots, O_n)$ for which there exists some $\alpha>1$ that supports it.
	\item There exists an algorithm that uses $poly(m,n)$ value queries that finds some allocation that is supported by some $\alpha>1$.
\end{enumerate}
\end{proposition}
\begin{proof}
Both proofs rely on applying Proposition \ref{prop-maximal}. For the first part, we show that there exists a welfare maximizing allocation that is maximal: start with some optimal allocation $(O_1,\ldots, O_n)$ and consider the following process: if there is some player $i$ and item $j$ with $v_i(j|O_i)=0$, remove item $j$ from $O_i$. Repeat this process until obtaining an allocation $(O'_i,\ldots, O'_n)$ where every allocated item have a positive marginal value. Note that $(O'_i,\ldots, O'_n)$ has the same value as $(O_i,\ldots, O_n)$ and thus it is welfare maximizing.

Let $Z$ be the set of unallocated items. Observe that any other item has a posititve marginal contribution to the allocation. Furthermore, notice that for every $i$ we have that $v_i(Z|O'_i)=0$, since otherwise the welfare of the allocation $(O;_1,\ldots, O'_{i-1},O'_i+Z,O'_{i+1},\ldots, O'_n)$ is strictly larger than $OPT$ (this is a valid allocation, as items in $Z$ are not allocated at all), which is a contradiction. Now we can apply Proposition \ref{prop-maximal} and get that $(O'_i,\ldots, O'_n)$ can be supported by some $\alpha>1$.

For the second part of the proposition, we show how to efficiently find a maximal allocation. Start with the grand bundle $M$, and remove from it some item with a zero marginal value for player $1$: $v_1(j|M-\{j\}) =0$, if such exists. Then repeat this process for the next item that have a zero marginal value until obtaining some set $S_1$ such that for every $j\in S_1$, $v_1(j|S_1-\{j\}) >0$ and $v_1(S_1)=v_1(M)$. Similarly, starting with the remaining items $M-S_1$, obtain a bundle $S_2\subseteq M-S_1$ such that for every $j\in S_2$, $v_1(j|S_2-\{j\}) >0$ and $v_1(S_2)=v_1(M-S_1)$. Repeat similarly with the remaining items for the next players. Observe that $(S_1,\ldots, S_n)$ is maximal. Hence by the lemma there is some $\alpha>0$ that supports it. Note that the number of value queries that we make during the process is $poly(n,m)$.
\end{proof}

Not only are there allocations which cannot be supported by any $\alpha$, even some \emph{welfare maximizing} allocations cannot be supported by any $\alpha$. Specifically, this is the case for some welfare maximizing allocations in which items of zero marginal value are allocated to agents. See Example \ref{example:opt-not-supported} below. Note that this is very different than in the case of Walrasian equilibrium (when $\alpha=1$) in which allocating item of zero marginal value to the agents, and pricing them at zero, cannot destroy the equilibrium. 

\begin{example}\label{example:opt-not-supported}	
Consider a setting with three items and two players, where the value of the players for a bundle is only a function of its size: for player $1$, the value of any pair of items (and by monotonicity also of the grand bundle) is $1$, and for player $2$ the value of any pair of items is $\epsilon$ satisfying $1>\epsilon>0$. An optimal allocation is to allocate all three items to player $1$. In this case, the marginal value of each of the items is $0=v_1(1|2)$. Thus the price of each of the items must be $0$ since if it is positive the profit of player $1$ for a bundle that contains two items is bigger than the profit of his equilibrium allocation. However, for a price of $0$ for each item, any pair of items give a positive profit for player $2$, which is bigger than his $0$ profit from the empty bundle.
	
Note that in this example the welfare maximizing allocation that allocates items $\{a,b\}$ to player $1$ and leaves item $c$ unallocated is supported by, e.g., $p_a=1, p_b=1, p_c=0$ and $\alpha=2$.
\end{example}

\section{$\alpha$-Endowed Equilibrium for $\alpha<1$}\label{appendix-small-alpha}

In this section we discuss the case of $\alpha<1$. 
We show that if an $\alpha$-endowed equilibrium exists then its allocation must maximize the welfare with respect to the original valuations. 
This implies that to be able to present such an equilibrium, we must be able to compute a welfare maximizing allocation. 
Additionally, we show that even for unit-demand valuations, an $\alpha$-endowed equilibrium might fail to exist for any  $0<\alpha<1$, establishing that for gross-substitutes valuations the minimal $\alpha$ needed to  always support an $\alpha$-endowed equilibrium is indeed $1$.

\subsection{Any $\alpha$-Endowed Equilibrium for $0<\alpha<1$ is Welfare Maximizing}
We first show that for $0<\alpha\leq 1$ any $\alpha$-supported allocation is socially efficient. 

\begin{theorem}
	For any $0<\alpha\leq 1$, if an $\alpha$-endowed equilibrium exists then its allocation is welfare maximizing  with respect to the original valuations.
\end{theorem}
\begin{proof}
Consider an endowed equilibrium with allocation $(S_1,\ldots, S_n)$ with prices $p_1,\ldots, p_m$. Let $(O_1,\ldots, O_n)$ be a welfare maximizing allocation. For each player $i$ we have:
\begin{align*}
\alpha \cdot v_i(S_i) - p(S_i) 
&\geq \alpha \cdot v_i(S_i \cap O_i) + v_i(O_i - S_i | S_i \cap O_i) - \sum_{j\in O_i}p_j \\
&= \alpha \cdot v_i(S_i \cap O_i) + v_i(O_i) - v_i(S_i \cap O_i)- \sum_{j\in O_i}p_j\\
&= v_i(O_i) - (1-\alpha) v_i(S_i \cap O_i)- \sum_{j\in O_i}p_j\\
&\geq v_i(O_i) - (1-\alpha) v_i(O_i)- \sum_{j\in O_i}p_j
= \alpha \cdot v_i(O_i) - \sum_{j\in O_i}p_j
\end{align*}
Where the last inequality is due to the fact that for all $i$ it holds that $v_i(O_i)\geq v_i(S_i \cap O_i)$, and since $1-\alpha\geq 0$. Taking a sum over all the players and using $\sum_i\sum_{j\in S_i} p_j = \sum_i\sum_{j\in O_i} p_j$ we get that:
\begin{align*}
\alpha \sum_i v_i(S_i) \geq \alpha \sum_i v_i(O_i)
\end{align*}
We conclude that $\sum_i v_i(S_i)\geq \sum_i v_i(O_i)$, so $(S_1,\ldots, S_n)$ is socially efficient.
\end{proof}

\subsection{Any $\alpha$-Endowed Equilibrium for $\alpha=0$ is Welfare Maximizing}
We next show that for $\alpha=0$ any $\alpha$-supported allocation is socially efficient. 
\begin{observation}
	For $\alpha=0 $, the allocation of any $\alpha$-endowed equilibrium is welfare maximizing with respect to the original valuations.
\end{observation}
\begin{proof}
Fix some $\alpha$-endowed equilibrium $S= (S_1,\ldots, S_n)$ for $\alpha=0$. 
Since $\alpha=0$, each player $i$ has no endowed value for the items he gets, so he must be paying 0 for each of the items in $S_i$ (otherwise he prefers dropping all items with positive price). Thus the price of all items must be $0$ and the endowed profit is zero: $v_i^{S_i,\alpha}(S_i)-\Sigma_{j\in S_i} p_j=0\cdot v_i(S_i)-0=0$. 
As $i$ prefers $S_i$ to $M$ it holds that $v_i^{S_i,\alpha}(S_i)-\Sigma_{j\in S_i} p_j=0\geq 0\cdot v_i(S_i) + v_i(M|S_i) - \sum_{j\in M}p_j = v_i(M|S_i)- 0=v_i(M) - v_i(S_i)\geq 0$, and thus $v_i(M) = v_i(S_i)$.
 
We get that $(S_1,\ldots, S_n)$ is an $0$-endowed equilibrium if and only if $v_i(M) = v_i(S_i)$ for player $i$.
We next show that this implies that $S$ is welfare maximizing. Assume that it is not, and there is an allocation $(O_1, O_2,\ldots, O_n)$ such that
$\sum_{i} v_i(O_i)> \sum_{i} v_i(S_i)$. For this to be possible, it must be the case that for at least one player $i$ it holds that $v_i(O_i)>v_i(S_i)$. But this yields a contradiction as 
$v_i(M)\geq v_i(O_i)>v_i(S_i)= v_i(M)$.
We conclude that if $(S_1,\ldots, S_n)$ is the allocation of an $\alpha$-endowed equilibrium for $\alpha=0$ then it is welfare maximizing with respect to the original valuations.
\end{proof}

\subsection{Unit-demand Valuations are not $\alpha$-supported for $0<\alpha<1$}
\label{subsec-alpha-less-1}
Recall that an $\alpha$-endowed equilibrium for $\alpha=1$ is simply a Walrasian equilibrium. When valuations are gross-substitutes (e.g. unit demand) then it is well known that a Walrasian equilibrium exists. We next show that even for unit-demand valuations, an $\alpha$-endowed equilibrium does not exist for $\alpha<1$, and thus for unit-demand valuations it is indeed required that $\alpha\geq 1$ to ensure the existence of an $\alpha$-endowed equilibrium. 
	
When $\alpha<1$ each player discounts the value of the items that he receives. Here is a very simple instances with unit-demand valuations that does not admit any $\alpha$-endowed equilibrium for $\alpha<1$: consider $n$ identical players and $n$ identical items. Each player is unit demand and wants any single item for a value of $1$. Observe that if there is an $\alpha$-endowed equilibrium, then by individual rationality the price of each item is be at most $\alpha$. Let player $i$ be the player that is allocated the item $j$ with the highest price $p_j$ (if there are several such players, choose one arbitrarily). Observe that player $i$ prefers any other item $j'$: $\alpha\cdot  v_i(\{j\}) - p_j <  v (\{j'\}) - p_{j'}$. This is because $\alpha < 1$ and since $p_i \geq p_j$. Thus, no $\alpha$-endowed equilibrium exists in this instance.

\section{The Integrality gap of $2$-Player Instances with Subaddititve Valuations}\label{app-subaddititve}

We now show that the integrality gap of instances with two players, both with subadditive valuations is strictly less than $2$. We consider some fractional solution $\{x_{i,S}\}$ of the LP and show how to round it to an integral solution that {provides an approximation ratio better than $2$.} We will prove the that there is an integral allocation $(S_1,S_2)$ such that $v_1(S_1)+v_2(S_2)\geq (\frac 1 2 + \frac 1 {2m})\Sigma_{i=1}^2\Sigma_{S}x_{i,S}v_i(S)$, as needed.

Suppose, without loss of generality, that $\Sigma_{S}x_{1,S}v_1(S)\geq \Sigma_{S}x_{2,S}v_2(S)$.  Sample a bundle $S_1$ by the distribution that assigned probability $x_{1,S}$ to each bundle $S$. Allocate $S_1$ to bidder $1$ and $M-S_1$ to bidder $2$. The expected value of this rounded solution is exactly $\Sigma_Sx_{1,S}\cdot (v_1(S)+v_2(M-S))$. It is therefore enough to prove that $\Sigma_Sx_{1,S}\cdot v_2(M-S)\geq \frac {\Sigma_Sx_{2,S}v_2(S)} {m}$.  This is so as this implies that the value of the solution is at least $\Sigma_Sx_{1,S}v_1(S)+\frac {\Sigma_Sx_{2,S}v_2(S)} {m}$ and $\Sigma_Sx_{1,S}v_1(S)\geq \Sigma_Sx_{2,S}v_2(S)$ so overall the expected value is at least $(\frac 1 2 + \frac 1 {2m})\Sigma_{i=1}^2\Sigma_{S}x_{i,S}v_i(S)$, as needed.

Given player $i$ and item $j$, let $q^i_j=\Sigma_{S:j\in S}x_{i,S}$. In particular, for any item $j$ we have that $q^1_j+q^2_j\leq 1$. Observe that the probability that player $2$ receives item $j$ is $1-q^1_j\geq q^2_j$. That is, let $\mathcal D$ be the distribution in which player $2$ receives bundle $S$ with probability $x_{2,S}$ and $\mathcal D'$ be the distribution where player $2$ receives bundle $M-S$ with probability $x_{1,S}$. The marginal distribution of player $2$ receiving item $j$ in $\mathcal D'$ is at least the marginal distribution of player $2$ receiving item $j$ in $\mathcal D$. We have that:
\begin{align*}
E_{S\sim \mathcal D'}[v_2(S)] &\geq \frac {E_{S\sim \mathcal D'}[\Sigma_{j\in S}v_2(\{j\})]} m=\frac {\Sigma_j(1-q^1_j)\cdot v_2(\{j\})} m \\
&\geq \frac {\Sigma_jq^2_j\cdot v_2(\{j\})} m \geq \frac {\Sigma_Sx_{2,S}v_2(S)} m =\frac {E_{S\sim \mathcal D}[v_2(S)]} m
\end{align*}

where in the first and last inequalities we use the subadditivity of $v_2$.

\section{The Communication Complexity of Finding a Local Maximum}\label{appendix-comm-hardness}

The paper \cite{BDN18} studies the following communication variant of local search: let $G=(V,E)$ be a known graph. Alice holds a function $f_A:V\rightarrow \mathbb R$ and Bob holds a function $f_B:V\rightarrow \mathbb{R}$. The goal is to find a local maximum of the graph: a vertex $v$ such that $f_A(v)+f_B(v)\geq f_A(u)+f_B(u)$ for any neighbor $u$ of $v$. In particular, it is shown in \cite{BDN18} that if $G$ is the odd graph then finding a local maximum requires $|V|^c$ bits of communication, for some constant $c$. This results holds even for randomized protocols.

We utilize this communication hardness result to prove that the communication complexity of finding a local maximum in combinatorial auctions is $exp(m)$. The proof is similar to the proof of Claim \ref{claim-value-hardness}. Again, let the number of items be $m=2k+1$ for some integer $k\geq 1$. The valuation of each player $i$ belongs to the following family:

\begin{align*}
	v_i(S)=\left\{
	\begin{array}{ll}
	|S|, & |S|\leq k-1\\
	k-1+\frac 1 2+c^i_{M-S}, & |S|=k \\
	k-1+\frac 3 {4}+c^i_S, & |S|=k+1 \\
	k, & |S|\geq k+2.
	\end{array}
	\right.
\end{align*}
It is not hard to see that if for every $S$ it holds that $c^i_S\leq \frac 1 {4}$ then the valuations are monotone and submodular. Also, similarly to before, an allocation $(S,M-S)$ can be a local maximum if either $|S|=k$ or $|M-S|=k$.

In our reduction from the communication hardness of local search on the $k$'th odd graph, we have $m=2k+1$ items. We let the valuation of Alice belong to the family above with $c^A_S=f_A(S)$.
Bob's valuation is similarly defined with $c^B_S=f_B(S)$. The proof that every local maximum of the odd graph is a local maximum of the combinatorial auction is very similar to the corresponding part of the proof of Claim \ref{claim-value-hardness}. We proceed similarly as in the proof of Claim \ref{claim-value-hardness}, except that we use the fact that by construction of our valuations $v_1(S)+v_2(M-S)=v_1(M-S)+v_2(S)$. That is, the value of every allocation $(S,M-S)$ ($|S|=k$) is $2k-2+\frac 1 2 + \frac 3 4 + c^A_S+c^B_S$ and the value of every allocation $(M-S,S)$ ($|S|=k$) is also $2k-2+\frac 1 2 + \frac 3 4 + c^A_S+c^B_S$. Moving a single item $j$ increases the welfare if and only if $c^A_S+c^B_S\geq c^A_{(M-S)\cup\{j\}}+c^B_{(M-S)\cup\{j\}}$.

This establishes that $|V|^c=exp(m)$ bits of communication are needed to find a local maximum in a combinatorial auction with submodular bidders. 
}
\end{document}